\documentclass[runningheads]{llncs}

\usepackage[T1]{fontenc}

\usepackage{hyperref}

\usepackage{mathtools, stmaryrd}

\usepackage{tikz}
\usetikzlibrary{decorations.pathmorphing}

\usepackage{prettyref}
\newcommand{\rref}[2][]{\prettyref{#2}}
\newrefformat{sec}{Section\,\ref{#1}}
\newrefformat{app}{Appendix\,\ref{#1}}
\newrefformat{def}{Def.\,\ref{#1}}
\newrefformat{thm}{Theorem\,\ref{#1}}
\newrefformat{prop}{Proposition\,\ref{#1}}
\newrefformat{lem}{Lemma\,\ref{#1}}
\newrefformat{cor}{Corollary\,\ref{#1}}
\newrefformat{ex}{Example\,\ref{#1}}
\newrefformat{tab}{Table\,\ref{#1}}
\newrefformat{fig}{Fig.\,\ref{#1}}
\newrefformat{rem}{Remark\,\ref{#1}}
\newrefformat{case}{Case\,\ref{#1}}

\usepackage{mathpartir}

\usepackage[prefixflatinterpret,prefixflatbracketmodalinterpret,prefixflatsetinterpret,sidenotecalculus,longseqcontext,shortmquant,mquantifiertype,mconnectiveformal,irunistore]{logic}
\usepackage[prefixflatinterpret,prefixflatbracketmodalinterpret,differentialdL,precisenames,fixformat]{dL}
\usepackage{math}
\newcommand{\I}{\interpretation[state=\nu]}
\newcommand{\It}{\interpretation[state=\omega]}
\newcommand{\Itt}{\interpretation[state=\omega']}
\newcommand{\sol}{\varphi}
\newcommand{\Isol}[1][t]{\interpretation[state=\sol(#1), const=I]}

\usepackage{dLmacros}

\definecolor{vred}{rgb}{.7,0,0}
\definecolor{vblue}{rgb}{.1,.15,.62}
\definecolor{darkishgray}{rgb}{.35,.35,.35}

\renewcommand*{\irrulename}[1]{\text{\textcolor{darkishgray}{\upshape(#1)}}}

\newcommand{\projctx}{projective context}
\newcommand{\Projctx}{Projective context}
\begin{document}

\cinferenceRuleStore[leqtrans|$\refines_{t}$]{Transitivity of $\refines$}
{\linferenceRule[lpmi]
    {\axkey{\ausprg \refines \busprg}}
    {\ausprg \refines \cusprg \land \cusprg \refines \busprg}
}{}
\cinferenceRuleStore[leqantisym|$\prgeq$]{Antisymmetry of $\refines$}
{\linferenceRule[equiv]
    {\ausprg \refines \busprg \land \busprg \refines \ausprg}
    {\axkey{\ausprg \prgeq \busprg}}
}{}
\cinferenceRuleStore[boxleq|{[$\refines$]}]{Box refinement}
{\linferenceRule[impl]
    {\ausprg \refines \busprg}
    {(\axkey{\dbox{\ausprg}\ausfml}\lylpmi\dbox{\busprg}{\ausfml})}
}{}
\cinferenceRuleStore[test|?]{test axiom}
{\linferenceRule[equiv]
{(\osfml \limply \osfmlb)}
{(\axkey{\ptest{\osfml} \refines  {\ptest{\osfmlb}}})}
}{}

\cinferenceRuleStore[update|${:}{=}$]{Update axiom}
{\linferenceRule[eq]
    {\prandom{x};\ptest{x = \oustrm}}
    {\pupdate{\pumod{x}{\oustrm}}}
}{}

\cinferenceRuleStore[testdet|?$_{\text{det}}$]{deterministic test}
{\linferenceRule[equiv]
{\dbox{\ptest{\osfml}}{\ausprg \leq \busprg}}
{\axkey{\ptest{\osfml};\ausprg \leq \ptest{\osfml};\busprg}}
}{}
\cinferenceRuleStore[stutter|stutter]{stuttering}
{\linferenceRule[eq]
{{\ptest{\ltrue}}}
{\axkey{\pupdate{\pumod{x}{x}}}}
}{}
\cinferenceRuleStore[choicel|$\cup_l$]{choice left axiom}
{\linferenceRule[equiv]
{\ausprg \refines \cusprg \land \busprg \refines \cusprg}
{\axkey{\pchoice{\ausprg}{\busprg} \refines \cusprg}}
}{}

\cinferenceRuleStore[choicer|$\cup_r$]{choice right axiom}
{\linferenceRule[lpmi]
{\axkey{\ausprg \refines \pchoice{\busprg}{\cusprg}}}
{\ausprg \refines \busprg \lor \ausprg \refines \cusprg}
}{}

\cinferenceRuleStore[sequence|;]{sequence axiom}
{\linferenceRule[lpmi]
{\axkey{\ausprg;\busprg \refines \cusprg;\dusprg}}
{\ausprg \refines \cusprg \land \dbox{\ausprg}{\busprg \refines \dusprg}}
}{}

\cinferenceRuleStore[loopl|$\text{loop}_l$]{loop left axiom}
{\linferenceRule[lpmi]
{\axkey{\prepeat{\ausprg};\busprg \refines \busprg}}
{\dbox{\prepeat{\ausprg}}{\ausprg;\busprg \refines \busprg}}
}{}

\cinferenceRuleStore[loopr|$\text{loop}_r$]{loop right axiom}
{\linferenceRule[lpmi]
{\axkey{\ausprg;\prepeat{\busprg} \refines \ausprg}}
{\ausprg;\busprg \refines \ausprg}
}{}

\cinferenceRuleStore[unloop|$\text{unloop}$]{unloop axiom}
{\linferenceRule[lpmi]
    {\axkey{\prepeat{\ausprg}\refines\prepeat{\busprg}}}
    {\dbox{\prepeat{\ausprg}}(\ausprg\refines\busprg)}
}{}

\cinferenceRuleStore[assigndet|$:=_{\text{det}}$]{deterministic assignment}
{\linferenceRule[equiv]
{\dbox{\pupdate{\pumod{x}{\oustrm}}}{\ausprg \leq \busprg}}
{\axkey{\pupdate{\pumod{x}{\oustrm}};\ausprg \leq \pupdate{\pumod{x}{\oustrm}};\busprg}}
}{}	

\cinferenceRuleStore[randtestmerge|${:}*_{\text{merge}}$]{Merging randoms with test}
{\linferenceRule[eq]
    {\prandom{x};\ptest{\lexists y \ousfml[y]}}
    {\axkey{\prandom{x};\ptest{\ousfml[x]};\prandom{x}}}
}{}

\cinferenceRuleStore[randtestassignmerge|${:}{=}*_{\text{merge}}$]{Merging random and assignement with test}
{\linferenceRule[eq]
	{\prandom{x};\ptest{\lexists y (\ousfml[y] \land x = \oustrm[y])}}
	{\axkey{\prandom{x};\ptest{\ousfml[x]};\pupdate{\umod{x}{\oustrm[x]}}}}
}{}
	
\cinferenceRuleStore[ode|ODE]{ODE axiom}
{\linferenceRule[equiv]
  {\dbox{\pode{\D{x}=\oustrm[x]}{\ousfml[x]}}{(\D{x}=\oustrmb[x] \land \ousfmlb[x])}}
  {\axkey{\pode{\D{x}=\oustrm[x]}{\ousfml[x]}\refines\pode{\D{x}=\oustrmb[x]}{\ousfmlb[x]}}}
}{}

\cinferenceRuleStore[DWref|$\text{DW}_{\prgeq}$]{DW axiom}
{\linferenceRule[eq]
  {{\ptest{\ousfml[x]};\pode{\D{x}=\oustrm[x]}{\ousfml[x]};\ptest{\ousfml[x]}}}
  {\axkey{\pode{\D{x}=\oustrm[x]}{\ousfml[x]}}}
}{}

\cinferenceRuleStore[DEref|$\text{DE}_{\prgeq}$]{DE axiom}
{\linferenceRule[eq]
  {\pode{\D{x}=\oustrm[x]}{\ousfml[x]};\pupdate{\umod{\D{x}}{\oustrm[x]}}}
  {\axkey{\pode{\D{x}=\oustrm[x]}{\ousfml[x]}}}
}{}

\cinferenceRuleStore[DX|DX]{DX axiom}
{\linferenceRule[leq]
  {\axkey{\pode{\D{x}=\oustrm[x]}{\ousfml[x]}}}
  {\pupdate{\umod{\D{x}}{\oustrm[x]}};\ptest{\ousfml[x]}}
}{}

\cinferenceRuleStore[ODEidemp|ODE$_{\text{idemp}}$]{ODE idempotent axiom}
{\linferenceRule[eq]
  {\pode{\D{x}=\oustrm[x]}{\ousfml[x]}}
  {\axkey{\pode{\D{x}=\oustrm[x]}{\ousfml[x]};\pode{\D{x}=\oustrm[x]}{\ousfml[x]}}}
}{}


\cinferenceRuleStore[leqrefl|$\refines_{\text{refl}}$]{Reflexivity of $\refines$}
{\linferenceRule[leq]
    {\ausprg}
    {\ausprg}
}{}
{\cinferenceRuleStore[cupidemp|$\cup_{\text{idemp}}$]{Idempotence of $\cup$}
{\linferenceRule[eq]
    {\ausprg}
    {\axkey{\pchoice{\ausprg}{\ausprg}}}
}{}}
{\cinferenceRuleStore[cupassoc|$\cup_{\text{assoc}}$]{Associativity of $\cup$}
{\linferenceRule[eq]
    {\pchoice{\ausprg}{(\pchoice{\busprg}{\cusprg})}}
    {\pchoice{(\pchoice{\ausprg}{\busprg})}{\cusprg}}
}{}}
\cinferenceRuleStore[seqidl|$;_{\text{id-l}}$]{Left neutral of $;$}
{\linferenceRule[eq]
    {\ausprg}
    {\axkey{\ptest{\ltrue};\ausprg}}
}{}
\cinferenceRuleStore[seqdistl|dist-l]{Left distributivity of $;$}
{\linferenceRule[eq]
    {\ausprg;(\pchoice{\busprg}{\cusprg})}
    {\pchoice{(\ausprg;\busprg)}{(\ausprg;\cusprg)}}
}{}
\cinferenceRuleStore[seqannih-l|annih-l]{$;$ left annihilator}
{\linferenceRule[eq]
    {\ptest{\lfalse}}
    {\axkey{\ptest{\lfalse};\ausprg}}
}{}
{\cinferenceRuleStore[testand|$?_{\text{and}}$]{Test and}
{\linferenceRule[eq]
    {\ptest{\osfml \land \osfmlb}}
    {\axkey{\ptest{\osfml};\ptest{\osfmlb}}}
}{}}

\cinferenceRuleStore[unfold-l|unfold-l]{Left unfolding of $\prepeat{\ausprg}$}
{\linferenceRule[eq]
    {\prepeat{\ausprg}}
    {\pchoice{\ptest{\ltrue}}{(\ausprg;\prepeat{\ausprg})}}
}{}
\cinferenceRuleStore[cupid|$\cup_{\text{id}}$]{Neutral of $\cup$}
{\linferenceRule[eq]
    {\ausprg}
    {\axkey{\pchoice{\ausprg}{\ptest{\lfalse}}}}
}{}
{\cinferenceRuleStore[cupcomm|$\cup_{\text{comm}}$]{Commutativity of $\cup$}
{\linferenceRule[eq]
    {\pchoice{\ausprg}{\busprg}}
    {\pchoice{\busprg}{\ausprg}}
}{}}
\cinferenceRuleStore[seqassoc|$;_{\text{assoc}}$]{Associativity of $;$}
{\linferenceRule[eq]
    {(\ausprg;\busprg);\cusprg}
    {\ausprg;(\busprg;\cusprg)}
}{}
\cinferenceRuleStore[seqidr|$;_{\text{id-r}}$]{Right neutral of $;$}
{\linferenceRule[eq]
    {\ausprg}
    {\axkey{\ausprg;\ptest{\ltrue}}}
}{}
{\cinferenceRuleStore[seqdistr|dist-r]{Right distributivity of $;$}
{\linferenceRule[eq]
    {(\pchoice{\ausprg}{\busprg});\cusprg}
    {\pchoice{(\ausprg;\cusprg)}{(\busprg;\cusprg)}}
}{}}
\cinferenceRuleStore[seqannih-r|annih-r]{$;$ right annihilator}
{\linferenceRule[eq]
    {\ptest{\lfalse}}
    {\axkey{\ausprg;\ptest{\lfalse}}}
}{}
{\cinferenceRuleStore[testor|$?_{\text{or}}$]{Test or}
{\linferenceRule[eq]
    {\ptest{\osfml \lor \osfmlb}}
    {\axkey{\pchoice{\ptest{\osfml}}{\ptest{\osfmlb}}}}
}{}}
\cinferenceRuleStore[unfold-r|unfold-r]{Right unfolding of $\prepeat{\ausprg}$}
{\linferenceRule[eq]
    {\prepeat{\ausprg}}
    {\pchoice{\ptest{\ltrue}}{(\prepeat{\ausprg};\ausprg)}}
}{}


\cinferenceRuleStore[randcomm|${:}{}*_{\text{comm}}$]{Commutativity of randoms}
{\linferenceRule[eq]
    {\prandom{y};\prandom{x}}
    {\prandom{x};\prandom{y}}
}{}
{\cinferenceRuleStore[assignsub|$:=_{\text{sub}}$]{Substitution of assignments}
{\linferenceRule[eq]
{\pumod{x}{\oustrm};\pumod{y}{\oustrmb[\oustrm]}}
{\pumod{x}{\oustrm};\pumod{y}{\oustrmb[x]}}
}{}}
{\cinferenceRuleStore[assigntest|$:=_{\text{test}}$]{Commutativity of assignment and test}
{\linferenceRule[eq]
{\ptest{\ousfml[\oustrm]};\pumod{x}{\oustrm}}
{\pumod{x}{\oustrm};\ptest{\ousfml[x]}}
}{}}
{\cinferenceRuleStore[assignmerge|$:=_{\text{merge}}$]{Merging assignments}
{\linferenceRule[eq]
{\pumod{x}{\oustrmb[\oustrm]}}
{\axkey{\pumod{x}{\oustrm};\pumod{x}{\oustrmb[x]}}}
}{}}
{\cinferenceRuleStore[assigncomm|$:=_{\text{comm}}$]{Commutativity of assignments}
{\linferenceRule[eq]
{\pumod{y}{\oustrmb[\oustrm]};\pumod{x}{\oustrm}}
{\pumod{x}{\oustrm};\pumod{y}{\oustrmb[x]}}
}{}}
{\cinferenceRuleStore[assignrandcomm|${:}{=}*_{\text{comm}}$]{Commutativity of assignement and random}
{\linferenceRule[eq]
{\prandom{y};\pupdate{\pumod{x}{\oustrm}}}
{\pupdate{\pumod{x}{\oustrm}};\prandom{y}}
}{}}
\cinferenceRuleStore[randtest|${:}*_{\text{test}}$]{Commutativity of random and test}
{\linferenceRule[eq]
{\ptest{\osfml};\prandom{x}}
{\prandom{x};\ptest{\osfml}}
}{}


\providecommand{\axkey}[1]{\textcolor{vblue}{#1}}%
\cinferenceRuleStore[diamond|$\didia{\cdot}$]{diamond axiom}
{\linferenceRule[equiv]
  {\lnot\dbox{\ausprg}{\lnot \ausfml}}
  {\axkey{\ddiamond{\ausprg}{\ausfml}}}
}
{}
\cinferenceRuleStore[diamondax|$\didia{\cdot}$]{diamond axiom}
{\linferenceRule[equiv]
  {\lnot\dbox{\ausprgax}{\lnot \ausfmlax}}
  {\axkey{\ddiamond{\ausprgax}{\ausfmlax}}}
}
{}
\cinferenceRuleStore[assignb|$\dibox{:=}$]{assignment / substitution axiom}
{\linferenceRule[equiv]
  {p(\genDJ{x})}
  {\axkey{\dbox{\pupdate{\umod{x}{\genDJ{x}}}}{p(x)}}}
}
{}
\cinferenceRuleStore[assignbax|$\dibox{:=}$]{assignment / substitution axiom}
{\linferenceRule[equiv]
  {p(\aconst)}
  {\axkey{\dbox{\pupdate{\umod{x}{\aconst}}}{p(x)}}}
}
{}
\cinferenceRuleStore[Dassignb|$\dibox{:=}$]{differential assignment}
{\linferenceRule[equiv]
{p(\astrm)}
{\axkey{\dbox{\Dupdate{\Dumod{\D{x}}{\astrm}}}{p(\D{x})}}}
}
{}
\cinferenceRuleStore[testb|$\dibox{?}$]{test}
{\linferenceRule[equiv]
  {(\ivr \limply \ausfml)}
  {\axkey{\dbox{\ptest{\ivr}}{\ausfml}}}
}{}%
\cinferenceRuleStore[testbax|$\dibox{?}$]{test}
{\linferenceRule[equiv]
  {(q \limply p)}
  {\axkey{\dbox{\ptest{q}}{p}}}
}{}%
\cinferenceRuleStore[evolveb|$\dibox{'}$]{evolve}
{\linferenceRule[equiv]
  {\lforall{t{\geq}0}{\dbox{\pupdate{\pumod{x}{\solf(t)}}}{p(x)}}}
  {\axkey{\dbox{\pevolve{\D{x}=\genDE{x}}}{p(x)}}}
}{\m{\D{\solf}(t)=\genDE{\solf}}}
\cinferenceRuleStore[choiceb|$\dibox{\cup}$]{axiom of nondeterministic choice}
{\linferenceRule[equiv]
  {\dbox{\ausprg}{\ausfml} \land \dbox{\busprg}{\ausfml}}
  {\axkey{\dbox{\pchoice{\ausprg}{\busprg}}{\ausfml}}}
}{}%
\cinferenceRuleStore[choicebax|$\dibox{\cup}$]{axiom of nondeterministic choice}
{\linferenceRule[equiv]
  {\dbox{\ausprgax}{\ausfmlax} \land \dbox{\busprgax}{\ausfmlax}}
  {\axkey{\dbox{\pchoice{\ausprgax}{\busprgax}}{\ausfmlax}}}
}{}%
\cinferenceRuleStore[evolveinb|$\dibox{'}$]{evolve}
{\linferenceRule[equiv]
  {
        \lforall{t{\geq}0}{\big(
          (\lforall{0{\leq}s{\leq}t}{q(\solf(s))})
          \limply
          \dbox{\pupdate{\pumod{x}{\solf(t)}}}{p(x)}
        \big)}
      }
  {
        \dbox{\pevolvein{\D{x}=\genDE{x}}{q(x)}}{p(x)}
  }
}{}
\cinferenceRuleStore[composeb|$\dibox{{;}}$]{composition} 
{\linferenceRule[equiv]
  {\dbox{\ausprg}{\dbox{\busprg}{\ausfml}}}
  {\axkey{\dbox{\ausprg;\busprg}{\ausfml}}}
}{}%
\cinferenceRuleStore[composebax|$\dibox{{;}}$]{composition} 
{\linferenceRule[equiv]
  {\dbox{\ausprgax}{\dbox{\busprgax}{\ausfmlax}}}
  {\axkey{\dbox{\ausprgax;\busprgax}{\ausfmlax}}}
}{}%
\cinferenceRuleStore[iterateb|$\dibox{{}^*}$]{iteration/repeat unwind} 
{\linferenceRule[equiv]
  {\ausfml \land \dbox{\ausprg}{\dbox{\prepeat{\ausprg}}{\ausfml}}}
  {\axkey{\dbox{\prepeat{\ausprg}}{\ausfml}}}
}{}%
\cinferenceRuleStore[iteratebax|$\dibox{{}^*}$]{iteration/repeat unwind} 
{\linferenceRule[equiv]
  {\ausfmlax \land \dbox{\ausprgax}{\dbox{\prepeat{\ausprgax}}{\ausfmlax}}}
  {\axkey{\dbox{\prepeat{\ausprgax}}{\ausfmlax}}}
}{}%
\cinferenceRuleStore[K|K]{K axiom / modal modus ponens}
{\linferenceRule[impl]
  {\dbox{\ausprg}{(\ausfml\limply\busfml)}}
  {(\dbox{\ausprg}{\ausfml}\limply\axkey{\dbox{\ausprg}{\busfml}})}
}{}%
\cinferenceRuleStore[Kax|K]{K axiom / modal modus ponens}
{\linferenceRule[impl]
  {\dbox{\ausprgax}{(\ausfmlax\limply\busfmlax)}}
  {(\dbox{\ausprgax}{\ausfmlax}\limply\axkey{\dbox{\ausprgax}{\busfmlax}})}
}{}%
\cinferenceRuleStore[I|II]{loop induction}
{\linferenceRule[impl]
  {\dbox{\prepeat{\ausprg}}{(\ausfml\limply\dbox{\ausprg}{\ausfml})}}
  {(\ausfml\limply\axkey{\dbox{\prepeat{\ausprg}}{\ausfml}})}
}{}%
\cinferenceRuleStore[Ieq|I]{loop induction}
{\linferenceRule[equiv]
  {\ausfml \land \dbox{\prepeat{\ausprg}}{(\ausfml\limply\dbox{\ausprg}{\ausfml})}}
  {\axkey{\dbox{\prepeat{\ausprg}}{\ausfml}}}
}{}%
\cinferenceRuleStore[Ieqax|I]{loop induction}
{\linferenceRule[equiv]
  {\ausfmlax \land \dbox{\prepeat{\ausprgax}}{(\ausfmlax\limply\dbox{\ausprgax}{\ausfmlax})}}
  {\axkey{\dbox{\prepeat{\ausprgax}}{\ausfmlax}}}
}{}%
\dinferenceRuleStore[backiterateb|\usebox{\backiterateb}]{backwards iteration/repeat unwind}
{\linferenceRule[equiv]
  {\ausfml \land \dbox{\prepeat{\ausprg}}{\dbox{\ausprg}{\ausfml}}}
  {\axkey{\dbox{\prepeat{\ausprg}}{\ausfml}}}
}{}
\dinferenceRuleStore[iterateiterateb|$\dibox{{}^*{}^*}$]{double iteration}
{\linferenceRule[equiv]
  {\dbox{\prepeat{\ausprg}}{\ausfml}}
  {\axkey{\dbox{\prepeat{\ausprg};\prepeat{\ausprg}}{\ausfml}}}
}{}
\dinferenceRuleStore[iterateiterated|$\didia{{}^*{}^*}$]{double iteration}
{\linferenceRule[equiv]
  {\ddiamond{\prepeat{\ausprg}}{\ausfml}}
  {\axkey{\ddiamond{\prepeat{\ausprg};\prepeat{\ausprg}}{\ausfml}}}
}{}
\cinferenceRuleStore[B|B]{Barcan and converse}
{\linferenceRule[equiv]
        {\ddiamond{\ausprg}{\lexists{x}{\ausfml}}}
        {\lexists{x}{\ddiamond{\ausprg}{\ausfml}}}
}{\m{x{\not\in}\ausprg}}
\cinferenceRuleStore[V|V]{vacuous $\dbox{}{}$}
{\linferenceRule[impl]
  {p}
  {\axkey{\dbox{\ausprg}{p}}}
}{\m{FV(p)\cap BV(\ausprg)=\emptyset}}
\cinferenceRuleStore[Vax|V]{vacuous $\dbox{}{}$}
{\linferenceRule[impl]
  {p}
  {\axkey{\dbox{a}{p}}}
}{}
\cinferenceRuleStore[G|G]{$\dbox{}{}$ generalization} 
{\linferenceRule[formula]
  {\ausfml}
  {\dbox{\ausprg}{\ausfml}}
}{}%
\cinferenceRuleStore[Gax|G]{$\dbox{}{}$ generalization} 
{\linferenceRule[formula]
  {\ausfmlax}
  {\dbox{\ausprgax}{\ausfmlax}}
}{}%
\cinferenceRuleStore[genaax|$\forall{}$]{$\forall{}$ generalisation}
{\linferenceRule[formula]
  {p(x)}
  {\lforall{x}{p(x)}}
}{}
\cinferenceRuleStore[MPax|MP]{modus ponens}
{\linferenceRule[formula]
  {p\limply q \quad p}
  {q}
}{}
\cinferenceRuleStore[Mb|M${\dibox{\cdot}}$]{$\dbox{}{}$ monotone}
{\linferenceRule[formula]
  {\ausfml\limply \busfml}
  {\dbox{\ausprg}{\ausfml}\limply\dbox{\ausprg}{\busfml}}
}{}%
\cinferenceRuleStore[M|M]{$\ddiamond{}{}$ monotone / $\ddiamond{}{}$-generalization}
{\linferenceRule[formula]
  {\ausfml\limply\busfml}
  {\ddiamond{\ausprg}{\ausfml}\limply\ddiamond{\ausprg}{\busfml}}
}{}%

\dinferenceRuleStore[Mbr|M\rightrule]
{$\ddiamond{}{}/\dbox{}{}$ generalization=M=G+K} 
{\linferenceRule[sequent]
  {\lsequent[L]{} {\dbox{\ausprg}{\busfml}} 
  &\lsequent[g]{\busfml} {\ausfml}}
  {\lsequent[L]{} {\dbox{\ausprg}{\ausfml}}}
}{}%

\dinferenceRuleStore[loop|loop]{inductive invariant}
{\linferenceRule[sequent]
  {\lsequent[L]{} {\inv}
  &\lsequent[g]{\inv} {\dbox{\ausprg}{\inv}}
  &\lsequent[g]{\inv} {\ausfml}}
  {\lsequent[L]{} {\dbox{\prepeat{\ausprg}}{\ausfml}}}
}{}%
\dinferenceRuleStore[invind|ind]{inductive invariant}
{\linferenceRule[sequent]
  {\lsequent[\globalrule]{\ausfml}{\dbox{\ausprg}{\ausfml}}}
  {\lsequent{\ausfml}{\dbox{\prepeat{\ausprg}}{\ausfml}}}
}{}%
\cinferenceRuleStore[con|con]{loop convergence right} 
{\linferenceRule[formula]
  {\lsequent[G]{\mapply{\var}{v}\land v>0}{\ddiamond{\ausprg}{\mapply{\var}{v-1}}}}
  {\lsequent[L]{\lexists{v}{\mapply{\var}{v}}}
      {\axkey{\ddiamond{\prepeat{\ausprg}}{\lexists{v{\leq}0}{\mapply{\var}{v}}}}}}
}{v\not\in\ausprg}
\dinferenceRuleStore[congen|con]{loop convergence}
{\linferenceRule[sequent]
  {\lsequent[L]{}{\lexists{v}{\mapply{\var}{v}}}
  &\lsequent[G]{}{\lforall{v{>}0}{({\mapply{\var}{v}}\limply{\ddiamond{\ausprg}{\mapply{\var}{v-1}})}}}
  &\lsequent[G]{\lexists{v{\leq}0}{\mapply{\var}{v}}}{\busfml}
  }
  {\lsequent[L]{}{\ddiamond{\prepeat{\ausprg}}{\busfml}}}
}{v\not\in\ausprg}

\dinferenceRuleStore[band|${[]\land}$]{$\dbox{\cdot}{\land}$}
{\linferenceRule[equiv]
  {\dbox{\ausprg}{\ausfml} \land \dbox{\ausprg}{\busfml}}
  {\axkey{\dbox{\ausprg}{(\ausfml\land\busfml)}}}
}{}%

\dinferenceRuleStore[Hoarecompose|H${;}$]{Hoare $;$}
{\linferenceRule
  {A\limply\dbox{\ausprg}{E} & E\limply\dbox{\busprg}{B}}
  {A \limply \dbox{\ausprg;\busprg}{B}}
}{}%
\dinferenceRuleStore[composebrexplicit|$\dibox{{;}}$\rightrule]{$;$}
{\linferenceRule
  {A\limply\dbox{\ausprg}{\dbox{\busprg}{B}}}
  {A \limply \dbox{\ausprg;\busprg}{B}}
}{}


\cinferenceRuleStore[notr|$\lnot$\rightrule]{$\lnot$ right}
{\linferenceRule[sequent]
  {\lsequent[L]{\asfml}{}}
  {\lsequent[L]{}{\lnot \asfml}}
}{}%
\cinferenceRuleStore[notl|$\lnot$\leftrule]{$\lnot$ left}
{\linferenceRule[sequent]
  {\lsequent[L]{}{\asfml}}
  {\lsequent[L]{\lnot \asfml}{}}
}{}%
\cinferenceRuleStore[andr|$\land$\rightrule]{$\land$ right}
{\linferenceRule[sequent]
  {\lsequent[L]{}{\asfml}
    & \lsequent[L]{}{\bsfml}}
  {\lsequent[L]{}{\asfml \land \bsfml}}
}{}%
\cinferenceRuleStore[andl|$\land$\leftrule]{$\land$ left}
{\linferenceRule[sequent]
  {\lsequent[L]{\asfml , \bsfml}{}}
  {\lsequent[L]{\asfml \land \bsfml}{}}
}{}%
\cinferenceRuleStore[orr|$\lor$\rightrule]{$\lor$ right}
{\linferenceRule[sequent]
  {\lsequent[L]{}{\asfml, \bsfml}}
  {\lsequent[L]{}{\asfml \lor \bsfml}}
}{}%
\cinferenceRuleStore[orl|$\lor$\leftrule]{$\lor$ left}
{\linferenceRule[sequent]
  {\lsequent[L]{\asfml}{}
    & \lsequent[L]{\bsfml}{}}
  {\lsequent[L]{\asfml \lor \bsfml}{}}
}{}%
\cinferenceRuleStore[implyr|$\limply$\rightrule]{$\limply$ right}
{\linferenceRule[sequent]
  {\lsequent[L]{\asfml}{\bsfml}}
  {\lsequent[L]{}{\asfml \limply \bsfml}}
}{}%
\cinferenceRuleStore[implyl|$\limply$\leftrule]{$\limply$ left}
{\linferenceRule[sequent]
  {\lsequent[L]{}{\asfml}
    & \lsequent[L]{\bsfml}{}}
  {\lsequent[L]{\asfml \limply \bsfml}{}}
}{}%
\cinferenceRuleStore[equivr|$\lbisubjunct$\rightrule]{$\lbisubjunct$ right}
{\linferenceRule[sequent]
  {\lsequent[L]{\asfml}{\bsfml}
   & \lsequent[L]{\bsfml}{\asfml}}
  {\lsequent[L]{}{\asfml \lbisubjunct \bsfml}}
}{}%
\cinferenceRuleStore[equivl|$\lbisubjunct$\leftrule]{$\lbisubjunct$ left}
{\linferenceRule[sequent]
  {\lsequent[L]{\asfml\limply\bsfml, \bsfml\limply\asfml}{}}
  {\lsequent[L]{\asfml \lbisubjunct \bsfml}{}}
}{}%
\cinferenceRuleStore[id|id]{identity}
{\linferenceRule[sequent]
  {}
  {\lsequent[L]{\asfml}{\asfml}}
}{}%
\cinferenceRuleStore[cut|cut]{cut}
{\linferenceRule[sequent]
  {\lsequent[L]{}{\cusfml}
  &\lsequent[L]{\cusfml}{}}
  {\lsequent[L]{}{}}
}{}%
\cinferenceRuleStore[weakenr|W\rightrule]{weakening right}
{\linferenceRule[sequent]
  {\lsequent[L]{}{}}
  {\lsequent[L]{}{\asfml}}
}{}%
\cinferenceRuleStore[weakenl|W\leftrule]{weakening left}
{\linferenceRule[sequent]
  {\lsequent[L]{}{}}
  {\lsequent[L]{\asfml}{}}
}{}%
\cinferenceRuleStore[exchanger|P\rightrule]{exchange right}
{\linferenceRule[sequent]
  {\lsequent[L]{}{\bsfml,\asfml}}
  {\lsequent[L]{}{\asfml,\bsfml}}
}{}%
\cinferenceRuleStore[exchangel|P\leftrule]{exchange left}
{\linferenceRule[sequent]
  {\lsequent[L]{\bsfml,\asfml}{}}
  {\lsequent[L]{\asfml,\bsfml}{}}
}{}%
\cinferenceRuleStore[contractr|c\rightrule]{contract right}
{\linferenceRule[sequent]
  {\lsequent[L]{}{\asfml,\asfml}}
  {\lsequent[L]{}{\asfml}}
}{}%
\cinferenceRuleStore[contractl|c\leftrule]{contract left}
{\linferenceRule[sequent]
  {\lsequent[L]{\asfml,\asfml}{}}
  {\lsequent[L]{\asfml}{}}
}{}
\cinferenceRuleStore[closeTrue|$\top$\rightrule]{close by true}
{\linferenceRule[sequent]
  {}
  {\lsequent[L]{}{\ltrue}}
}{}%
\cinferenceRuleStore[closeFalse|$\bot$\leftrule]{close by false}
{\linferenceRule[sequent]
  {}
  {\lsequent[L]{\lfalse}{}}
}{}%

\cinferenceRuleStore[CE|CE]{congequiv congruence of equivalences on formulas}
{\linferenceRule[formula]
  {\ausfml \lbisubjunct \busfml}
  {\contextapp{C}{\ausfml} \lbisubjunct \contextapp{C}{\busfml}}
}{}%
\dinferenceRuleStore[CEr|CE\rightrule]{congequiv congruence of equivalences on formulas}
{\linferenceRule[formula]
  {\lsequent[L]{} {\contextapp{C}{\busfml}}
  &\lsequent[g]{} {\ausfml \lbisubjunct \busfml}}
  {\lsequent[L]{} {\contextapp{C}{\ausfml}}}
}{}%
\dinferenceRuleStore[CEl|CE\leftrule]{congequiv congruence of equivalences on formulas}
{\linferenceRule[formula]
  {\lsequent[L]{\contextapp{C}{\busfml}} {}
  &\lsequent[g]{} {\ausfml \lbisubjunct \busfml}}
  {\lsequent[L]{\contextapp{C}{\ausfml}} {}}
}{}%

\cinferenceRuleStore[allr|$\forall$\rightrule]{$\lforall{}{}$ right}
{\linferenceRule[sequent]
  {\lsequent[L]{}{p(y)}}
  {\lsequent[L]{}{\lforall{x}{p(x)}}}
}{\m{y\not\in\Gamma{,}\Delta{,}\lforall{x}{p(x)}}}%
\cinferenceRuleStore[alll|$\forall$\leftrule]{$\lforall{}{}$ left instantiation}
{\linferenceRule[sequent]
  {\lsequent[L]{p(\astrm)}{}}
  {\lsequent[L]{\lforall{x}{p(x)}}{}}
}{arbitrary term $\astrm$}
\cinferenceRuleStore[existsr|$\exists$\rightrule]{$\lexists{}{}$ right}
{\linferenceRule[sequent]
  {\lsequent[L]{}{p(\astrm)}}
  {\lsequent[L]{}{\lexists{x}{p(x)}}}
}{arbitrary term $\astrm$}
\cinferenceRuleStore[existsl|$\exists$\leftrule]{$\lexists{}{}$ left}
{\linferenceRule[sequent]
  {\lsequent[L]{p(y)} {}}
  {\lsequent[L]{\lexists{x}{p(x)}} {}}
}{\m{y\not\in\Gamma{,}\Delta{,}\lexists{x}{p(x)}}}%

\cinferenceRuleStore[qear|\usebox{\Rval}]{quantifier elimination real arithmetic}
{\linferenceRule[sequent]
  {}
  {\lsequent[g]{\Gamma}{\Delta}}
}{\text{if}~\landfold_{\ausfml\in\Gamma} \ausfml \limply \lorfold_{\busfml\in\Delta} \busfml ~\text{is valid in \LOS[\reals]}}%

\dinferenceRuleStore[allGi|i$\forall$]{inverse universal generalization / universal instantiation}
{\linferenceRule[sequent]
  {\lsequent[L]{} {\lforall{x}{\ausfml}}}
  {\lsequent[L]{} {\ausfml}}
}{}

\cinferenceRuleStore[applyeqr|=\rightrule]{apply equation}
{\linferenceRule[sequent]
  {\lsequent[L]{x=\astrm}{p(\astrm)}}
  {\lsequent[L]{x=\astrm}{p(x)}}
}{}%
\cinferenceRuleStore[applyeql|=\leftrule]{apply equation}
{\linferenceRule[sequent]
  {\lsequent[L]{x=\astrm,p(\astrm)}{}}
  {\lsequent[L]{x=\astrm,p(x)}{}}
}{}%

\dinferenceRuleStore[alldupl|$\forall\forall$\leftrule]{$\lforall{}{}$ left instantiation retaining duplicates}
{\linferenceRule[sequent]
  {\lsequent[L]{\lforall{x}{p(x)},p(\astrm)}{}}
  {\lsequent[L]{\lforall{x}{p(x)}}{}}
}{}

\dinferenceRuleStore[choicebrinsist|$\dibox{\cup}\rightrule$]{}
{\linferenceRule
  {\lsequent[L]{}{\dbox{\asprg}{\ausfml}\land\dbox{\bsprg}{\ausfml}}}
  {\lsequent[L]{}{\dbox{\pchoice{\asprg}{\bsprg}}{\ausfml}}}
}{}
\dinferenceRuleStore[choiceblinsist|$\dibox{\cup}\leftrule$]{}
{\linferenceRule
  {\lsequent[L]{\dbox{\asprg}{\ausfml}\land\dbox{\bsprg}{\ausfml}}{}}
  {\lsequent[L]{\dbox{\pchoice{\asprg}{\bsprg}}{\ausfml}}{}}
}{}
\dinferenceRuleStore[choicebrinsist2|$\dibox{\cup}\rightrule2$]{}
{\linferenceRule
  {\lsequent[L]{}{\dbox{\asprg}{\ausfml}}
  &\lsequent[L]{}{\dbox{\bsprg}{\ausfml}}}
  {\lsequent[L]{}{\dbox{\pchoice{\asprg}{\bsprg}}{\ausfml}}}
}{}
\dinferenceRuleStore[choiceblinsist2|$\dibox{\cup}\leftrule2$]{}
{\linferenceRule
  {\lsequent[L]{\dbox{\asprg}{\ausfml},\dbox{\bsprg}{\ausfml}}{}}
  {\lsequent[L]{\dbox{\pchoice{\asprg}{\bsprg}}{\ausfml}}{}}
}{}
\dinferenceRuleStore[cutr|cut\rightrule]{cut right}
{\linferenceRule[sequent]
  {\lsequent[L]{}{\bsfml}
  &\lsequent[L]{}{\bsfml\limply\asfml}}
  {\lsequent[L]{}{\asfml}}
}{}
\dinferenceRuleStore[cutl|cut\leftrule]{cut left}
{\linferenceRule[sequent]
  {\lsequent[L]{\bsfml} {}
  &\lsequent[L]{}{\asfml\limply\bsfml}}
  {\lsequent[L]{\asfml} {}}
}{}


\cinferenceRuleStore[Dplus|$+'$]{derive sum}
{\linferenceRule[eq]
  {\der{\asdtrm}+\der{\bsdtrm}}
  {\axkey{\der{\asdtrm+\bsdtrm}}}
}
{}
\cinferenceRuleStore[Dplusax|$+'$]{derive sum}
{\linferenceRule[eq]
  {\der{\asdtrmax}+\der{\bsdtrmax}}
  {\axkey{\der{\asdtrmax+\bsdtrmax}}}
}
{}
\cinferenceRuleStore[Dminus|$-'$]{derive minus}
{\linferenceRule[eq]
  {\der{\asdtrm}-\der{\bsdtrm}}
  {\axkey{\der{\asdtrm-\bsdtrm}}}
}
{}
\cinferenceRuleStore[Dminusax|$-'$]{derive minus}
{\linferenceRule[eq]
  {\der{\asdtrmax}-\der{\bsdtrmax}}
  {\axkey{\der{\asdtrmax-\bsdtrmax}}}
}
{}
\cinferenceRuleStore[Dtimes|$\cdot'$]{derive product}
{\linferenceRule[eq]
  {\der{\asdtrm}\cdot \bsdtrm+\asdtrm\cdot\der{\bsdtrm}}
  {\axkey{\der{\asdtrm\cdot \bsdtrm}}}
}
{}
\cinferenceRuleStore[Dtimesax|$\cdot'$]{derive product}
{\linferenceRule[eq]
  {\der{\asdtrmax}\cdot \bsdtrmax+\asdtrmax\cdot\der{\bsdtrmax}}
  {\axkey{\der{\asdtrmax\cdot \bsdtrmax}}}
}
{}
\cinferenceRuleStore[Dquotient|$/'$]{derive quotient}
{\linferenceRule[eq]
  {\big(\der{\asdtrm}\cdot \bsdtrm-\asdtrm\cdot\der{\bsdtrm}\big) / \bsdtrm^2}
  {\axkey{\der{\asdtrm/\bsdtrm}}}
}
{}
\cinferenceRuleStore[Dquotientax|$/'$]{derive quotient}
{\linferenceRule[eq]
  {\big(\der{\asdtrmax}\cdot \bsdtrmax-\asdtrmax\cdot\der{\bsdtrmax}\big) / \bsdtrmax^2}
  {\axkey{\der{\asdtrmax/\bsdtrmax}}}
}
{}
\cinferenceRuleStore[Dconst|$c'$]{derive constant}
{\linferenceRule[eq]
  {0}
  {\axkey{\der{\aconst}}}
  \hspace{3cm}
}
{\text{for numbers or constants~$\aconst$}}
\cinferenceRuleStore[Dvar|$x'$]{derive variable}
{\linferenceRule[eq]
  {\D{x}}
  {\axkey{\der{x}}}
}
{\text{for variable~$x\in\allvars$}}

\cinferenceRuleStore[DE|DE]{differential effect} 
{\linferenceRule[viuqe]
  {\axkey{\dbox{\pevolvein{\D{x}=\genDE{x}}{\ivr}}{\ausfml}}}
  {\dbox{\pevolvein{\D{x}=f(x)}{\ivr}}{\dbox{\axeffect{\Dupdate{\Dumod{\D{x}}{\genDE{x}}}}}{\ausfml}}}
}
{}%
\cinferenceRuleStore[DEax|DE]{differential effect} 
{\linferenceRule[viuqe]
  {\axkey{\dbox{\pevolvein{\D{x}=\genDE{x}}{q(x)}}{\ausfmlax}}}
  {\dbox{\pevolvein{\D{x}=f(x)}{q(x)}}{\dbox{\axeffect{\Dupdate{\Dumod{\D{x}}{\genDE{x}}}}}{\ausfmlax}}}
}
{}%


\cinferenceRuleStore[Dand|${\land}'$]{derive and}
{\linferenceRule[equiv]
  {\der{\asfml}\land\der{\bsfml}}
  {\axkey{\der{\asfml\land\bsfml}}}
}
{}
\cinferenceRuleStore[Dor|${\lor}'$]{derive or}
{\linferenceRule[equiv]
  {\der{\asfml}\land\der{\bsfml}}
  {\axkey{\der{\asfml\lor\bsfml}}}
}
{}
\cinferenceRuleStore[diffweaken|DW]{differential evolution domain} 
{\linferenceRule[viuqe]
  {\axkey{\dbox{\pevolvein{\D{x}=\genDE{x}}{\ivr}}{\ousfml[x]}}}
  {\dbox{\pevolvein{\D{x}=\genDE{x}}{\ivr}}{(\axeffect{\ivr}\limply \ousfml[x])}}
}
{}
\cinferenceRuleStore[diffweakenax|DW]{differential evolution domain} 
{\linferenceRule[viuqe]
  {\axkey{\dbox{\pevolvein{\D{x}=\genDE{x}}{q(x)}}{p(x)}}}
  {\dbox{\pevolvein{\D{x}=\genDE{x}}{q(x)}}{(\axeffect{q(x)}\limply p(x))}}
}
{}
\cinferenceRuleStore[dW|dW]{differential weakening}
{\linferenceRule[sequent]
  {\lsequent[g]{\ivr} {\ousfml[x]}}
  {\lsequent[g]{\Gamma} {\dbox{\pevolvein{\D{x}=f(x)}{\ivr}}{\ousfml[x]},\Delta}}
}
{}
\cinferenceRuleStore[DI|DI]{differential induction}
{\linferenceRule[lpmi]
  {\big(\axkey{\dbox{\pevolvein{\D{x}=\genDE{x}}{\ivr}}{\ousfml[x]}}
  \lbisubjunct \dbox{\ptest{\ivr}}{\ousfml[x]}\big)}
  {(\ivr\limply\dbox{\pevolvein{\D{x}=\genDE{x}}{\ivr}}{\axeffect{\der{\ousfml[x]}}})}
}
{}
\cinferenceRuleStore[DIax|DI]{differential induction}
{\linferenceRule[lpmi]
  {\big(\axkey{\dbox{\pevolvein{\D{x}=\genDE{x}}{q(x)}}{p(x)}}
  \lbisubjunct \dbox{\ptest{q(x)}}{p(x)}\big)}
  {(q(x)\limply\dbox{\pevolvein{\D{x}=\genDE{x}}{q(x)}}{\axeffect{\der{p(x)}}})}
}
{}
\cinferenceRuleStore[DIlight|DI]{differential induction}
{\linferenceRule[lpmi]
  {\big(\axkey{\dbox{\pevolvein{\D{x}=\genDE{x}}{\ivr}}{\ousfml[x]}}
  \lbisubjunct \dbox{\ptest{\ivr}}{\ousfml[x]}\big)}
  {\dbox{\pevolvein{\D{x}=\genDE{x}}{\ivr})}{\axeffect{\der{\ousfml[x]}}}}
}
{}

\cinferenceRuleStore[dI|dI]{differential invariant}
{\linferenceRule[sequent]
  {\lsequent[g]{\ivr}{\Dusubst{\D{x}}{\genDE{x}}{\der{F}}}}
  {\lsequent{F}{\dbox{\pevolvein{\D{x}=\genDE{x}}{\ivr}}{F}}}
}{}
\cinferenceRuleStore[DC|DC]{differential cut}
{\linferenceRule[lpmi]
  {\big(\axkey{\dbox{\pevolvein{\D{x}=\genDE{x}}{\ivr}}{\ousfml[x]}} \lbisubjunct \dbox{\pevolvein{\D{x}=\genDE{x}}{\ivr\land \axeffect{\ousfmlc[x]}}}{\ousfml[x]}\big)}
  {\dbox{\pevolvein{\D{x}=\genDE{x}}{\ivr}}{\axeffect{\ousfmlc[x]}}}
}
{}
\cinferenceRuleStore[DCax|DC]{differential cut}
{\linferenceRule[lpmi]
  {\big(\axkey{\dbox{\pevolvein{\D{x}=\genDE{x}}{q(x)}}{p(x)}} \lbisubjunct \dbox{\pevolvein{\D{x}=\genDE{x}}{q(x)\land \axeffect{r(x)}}}{p(x)}\big)}
  {\dbox{\pevolvein{\D{x}=\genDE{x}}{q(x)}}{\axeffect{r(x)}}}
}
{}
\cinferenceRuleStore[dC|dC]{differential cut}
{\linferenceRule[sequent]
  {\lsequent[L]{}{\dbox{\pevolvein{\D{x}=\genDE{x}}{\ivr}}{\axeffect{\cusfml}}}
  &\lsequent[L]{}{\dbox{\pevolvein{\D{x}=\genDE{x}}{(\ivr\land \axeffect{\cusfml})}}{\ousfml[x]}}}
  {\lsequent[L]{}{\dbox{\pevolvein{\D{x}=\genDE{x}}{\ivr}}{\ousfml[x]}}}
}{}
\cinferenceRuleStore[DGanyode|DG]{differential ghost variables (unsound!)}
{\linferenceRule[viuqe]
  {\axkey{\dbox{\pevolvein{\D{x}=\genDE{x}}{\ivr}}{\ousfml[x]}}}
  {\lexists{y}{\dbox{\pevolvein{\D{x}=\genDE{x}\syssep\axeffect{\D{y}=g(x,y)}}{\ivr}}{\ousfml[x]}}}
}
{}
\cinferenceRuleStore[DG|DG]{differential ghost variables}
{\linferenceRule[viuqe]
  {\axkey{\dbox{\pevolvein{\D{x}=\genDE{x}}{\ivr}}{\ousfml[x]}}}
  {\lexists{y}{\dbox{\pevolvein{\D{x}=\genDE{x}\syssep\axeffect{\D{y}=a(x)\cdot y+b(x)}}{\ivr}}{\ousfml[x]}}}
}
{}
\cinferenceRuleStore[DGax|DG]{differential ghost variables}
{\linferenceRule[viuqe]
  {\axkey{\dbox{\pevolvein{\D{x}=\genDE{x}}{q(x)}}{p(x)}}}
  {\lexists{y}{\dbox{\pevolvein{\D{x}=\genDE{x}\syssep\axeffect{\D{y}=a(x)\cdot y+b(x)}}{q(x)}}{p(x)}}}
}
{}
\cinferenceRuleStore[dG|dG]{dG}
{\linferenceRule[sequent]
  {\lsequent[L]{} {\lexists{y}{\dbox{\pevolvein{\D{x}=f(x)\syssep\axeffect{\D{y}=a(x)\cdot y+b(x)}}{\oivr[x]}}{\ousfml[x]}}}
  }
  {\lsequent[L]{} {\dbox{\pevolvein{\D{x}=f(x)}{\oivr[x]}}{\ousfml[x]}}}
}
{}%

\dinferenceRuleStore[assignbeqr|$\dibox{:=}_=$]{assignb}
  {\linferenceRule[sequent]
    {\lsequent[L]{y=\austrm} {p(y)}}
    {\lsequent[L]{} {\dbox{\pupdate{\umod{x}{\austrm}}}{p(x)}}}
  }
  {\text{$y$ new}}

\cinferenceRuleStore[DSax|DS]{(constant) differential equation solution} 
{\linferenceRule[viuqe]
  {\axkey{\dbox{\pevolvein{\D{x}=\aconst}{q(x)}}{p(x)}}}
  {\lforall{t{\geq}0}{\big((\lforall{0{\leq}s{\leq}t}{q(x+\aconst\itimes s)}) \limply \dbox{\pupdate{\pumod{x}{x+\aconst\itimes t}}}{p(x)}\big)}}
}
{}

\dinferenceRuleStore[DIeq0|DI]{differential invariant axiom}
{\linferenceRule[lpmi]
  {\big(\axkey{\dbox{\pevolve{\D{x}=\genDE{x}}}{\,\astrm=0}} \lbisubjunct \astrm=0\big)}
  {\dbox{\pevolve{\D{x}=\genDE{x}}}{\,\axeffect{\der{\astrm}=0}}}
}
{}
\dinferenceRuleStore[diffindeq0|dI]{differential invariant $=0$ case}
{\linferenceRule[sequent]
  {\lsequent{~}{\Dusubst{\D{x}}{\genDE{x}}{\der{\astrm}}=0}}
  {\lsequent{\astrm=0}{\dbox{\pevolve{\D{x}=\genDE{x}}}{\astrm=0}}}
}{}
\cinferenceRuleStore[Liec|dI$_c$]{}
{\linferenceRule
  {\lsequent{\ivr}{\Dusubst{\D{x}}{\genDE{x}}{\der{\astrm}}=0}}
  {\lsequent{}{\lforall{c}{\big(\astrm=c \limply \dbox{\pevolvein{\D{x}=\genDE{x}}{\ivr}}{\astrm=c}\big)}}}
}{}


\cinferenceRuleStore[DIeq|DI$_=$]{differential induction $=$ case}
{\linferenceRule[lpmi]
  {\big(\axkey{\dbox{\pevolvein{\D{x}=\genDE{x}}{\ivr}}{\asdtrm=\bsdtrm}}
  \lbisubjunct \dbox{\ptest{\ivr}}{\asdtrm=\bsdtrm}\big)}
  {\dbox{\pevolvein{\D{x}=\genDE{x}}{\ivr})}{\axeffect{\der{\asdtrm}=\der{\bsdtrm}}}}
}
{}

\dinferenceRuleStore[diffindgen|dI']{differential invariant}
{\linferenceRule[sequent]
  {\lsequent[L]{}{\inv}
  &\lsequent[g]{\ivr}{\Dusubst{\D{x}}{\genDE{x}}{\der{\inv}}}
  &\lsequent[g]{\inv}{\psi}
  }
  {\lsequent[L]{}{\dbox{\pevolvein{\D{x}=\genDE{x}}{\ivr}}{\psi}}}
}{}
\cinferenceRuleStore[diffindunsound|dI$_{??}$]{unsound}
{\linferenceRule[sequent]
  {\lsequent{\ivr\land\inv}{\Dusubst{\D{x}}{\genDE{x}}{\der{\inv}}}}
  {\lsequent{\inv}{\dbox{\pevolvein{\D{x}=\genDE{x}}{\ivr}}{\inv}}}
}{}

\dinferenceRuleStore[introaux|iG]{introduce discrete ghost variable}
{\linferenceRule[sequent]
  {\lsequent[L]{}{\dbox{\axeffect{\pupdate{\pumod{y}{\astrm}}}}{p}}}
  {\lsequent[L]{} {p}}
}{\text{$y$ new}}
\dinferenceRuleStore[diffaux|dA]{differential auxiliary variables}
{\linferenceRule[sequent]
  {\lsequent[\globalrule]{}{\inv\lbisubjunct\lexists{y}{G}}
  &\lsequent{G} {\dbox{\pevolvein{\D{x}=\genDE{x}\syssep\axeffect{\D{y}=a(x)\cdot y+b(x)}}{\ivr}}{G}}}
  {\lsequent{\inv} {\dbox{\pevolvein{\D{x}=\genDE{x}}{\ivr}}{\inv}}}
}{}
\cinferenceRuleStore[randomd|$\didia{{:}*}$]{nondeterministic assignment}
{\linferenceRule[equiv]
  {\lexists{x}{\ousfml[x]}}
  {\axkey{\ddiamond{\prandom{x}}{\ousfml[x]}}}
}{}
\cinferenceRuleStore[randomb|$\dibox{{:}*}$]{nondeterministic assignment}
{\linferenceRule[equiv]
  {\lforall{x}{\ousfml[x]}}
  {\axkey{\dbox{\prandom{x}}{\ousfml[x]}}}
}{}

\cinferenceRuleStore[box|$\dibox{\cdot}$]{box axiom}
{\linferenceRule[equiv]
  {\lnot\ddiamond{\ausprg}{\lnot\ausfml}}
  {\axkey{\dbox{\ausprg}{\ausfml}}}
}
{}
\cinferenceRuleStore[assignd|$\didia{:=}$]{assignment / substitution axiom}
{\linferenceRule[equiv]
  {p(\genDJ{x})}
  {\axkey{\ddiamond{\pupdate{\umod{x}{\genDJ{x}}}}{p(x)}}}
}
{}
\cinferenceRuleStore[evolved|$\didia{'}$]{evolve}
{\linferenceRule[equiv]
  {\lexists{t{\geq}0}{\ddiamond{\pupdate{\pumod{x}{\solf(t)}}}{p(x)}}\hspace{1cm}}
  {\axkey{\ddiamond{\pevolve{\D{x}=\genDE{x}}}{p(x)}}}
}{\m{\D{\solf}(t)=\genDE{\solf}}}
\cinferenceRuleStore[evolveind|$\didia{'}$]{evolve}
{\linferenceRule[equiv]
  {\lexists{t{\geq}0}{\big((\lforall{0{\leq}s{\leq}t}{q(\solf(s))}) \land 
  \ddiamond{\pupdate{\pumod{x}{\solf(t)}}}{p(x)}\big)}}
  {\axkey{\ddiamond{\pevolvein{\D{x}=\genDE{x}}{q(x)}}{p(x)}}}
}{\m{\D{\solf}(t)=\genDE{\solf}}}
\cinferenceRuleStore[testd|$\didia{?}$]{test}
{\linferenceRule[equiv]
  {\ivr \land \ausfml}
  {\axkey{\ddiamond{\ptest{\ivr}}{\ausfml}}}
}{}
\cinferenceRuleStore[choiced|$\didia{\cup}$]{axiom of nondeterministic choice}
{\linferenceRule[equiv]
  {\ddiamond{\ausprg}{\ausfml} \lor \ddiamond{\busprg}{\ausfml}}
  {\axkey{\ddiamond{\pchoice{\ausprg}{\busprg}}{\ausfml}}}
}{}
\cinferenceRuleStore[composed|$\didia{{;}}$]{composition}
{\linferenceRule[equiv]
  {\ddiamond{\ausprg}{\ddiamond{\busprg}{\ausfml}}}
  {\axkey{\ddiamond{\ausprg;\busprg}{\ausfml}}}
}{}
\cinferenceRuleStore[iterated|$\didia{{}^*}$]{iteration/repeat unwind pre-fixpoint, even fixpoint}
{\linferenceRule[equiv]
  {\ausfml \lor \ddiamond{\ausprg}{\ddiamond{\prepeat{\ausprg}}{\ausfml}}}
  {\axkey{\ddiamond{\prepeat{\ausprg}}{\ausfml}}}
}{}
\cinferenceRuleStore[duald|$\didia{{^d}}$]{dual}
{\linferenceRule[equiv]
  {\lnot\ddiamond{\ausprg}{\lnot\ausfml}}
  {\axkey{\ddiamond{\pdual{\ausprg}}{\ausfml}}}
}{}
\cinferenceRuleStore[dualb|$\dibox{{^d}}$]{dual}
{\linferenceRule[equiv]
  {\lnot\dbox{\ausprg}{\lnot\ausfml}}
  {\axkey{\dbox{\pdual{\ausprg}}{\ausfml}}}
}{}
\cinferenceRuleStore[FP|FP]{iteration is least fixpoint / reflexive transitive closure RTC, equivalent to invind in the presence of R}
{\linferenceRule[formula]
  {\ausfml \lor \ddiamond{\ausprg}{\busfml} \limply \busfml}
  {\ddiamond{\prepeat{\ausprg}}{\ausfml} \limply \busfml}
}{}
\cinferenceRuleStore[invindg|ind]{inductive invariant for games}
{\linferenceRule[formula]
  {\ausfml\limply\dbox{\ausprg}{\ausfml}}
  {\ausfml\limply\dbox{\prepeat{\ausprg}}{\ausfml}}
}{}

\dinferenceRuleStore[dchoiced|$\didia{{\cap}}$]{Demon's choice}
{
\axkey{\ddiamond{\dchoice{\ausprg}{\busprg}}{\ausfml}} \lbisubjunct \ddiamond{\ausprg}{\ausfml} \land \ddiamond{\busprg}{\ausfml}
}{}
\dinferenceRuleStore[dchoiceb|$\dibox{{\cap}}$]{Demon's choice}
{
\axkey{\dbox{\dchoice{\ausprg}{\busprg}}{\ausfml}} \lbisubjunct \dbox{\ausprg}{\ausfml} \lor \dbox{\busprg}{\ausfml}
}{}
\dinferenceRuleStore[diterateb|$\dibox{\drepeat{}}$]{Demon's repetition}
{\linferenceRule[equiv]
  {\ausfml \lor \dbox{\ausprg}{\dbox{\drepeat{\ausprg}}{\ausfml}}}
  {\axkey{\dbox{\drepeat{\ausprg}}{\ausfml}}}
}{}
\dinferenceRuleStore[diterated|$\didia{\drepeat{}}$]{Demon's repetition}
{\linferenceRule[equiv]
  {\ausfml \land \ddiamond{\ausprg}{\ddiamond{\drepeat{\ausprg}}}{\ausfml}}
  {\axkey{\ddiamond{\drepeat{\ausprg}}{\ausfml}}}
}{}
\dinferenceRuleStore[dinvindg|ind$\drepeat{}$]{inductive invariant for games}
{\linferenceRule[formula]
  {\ausfml\limply\ddiamond{\ausprg}{\ausfml}}
  {\ausfml\limply\ddiamond{\drepeat{\ausprg}}{\ausfml}}
}{}
\dinferenceRuleStore[dFP|FP$\drepeat{}$]{dual iteration is least fixpoint in Demon's winning strategy}
{\linferenceRule[formula]
  {\ausfml \lor \dbox{\ausprg}{\busfml} \limply \busfml}
  {\dbox{\drepeat{\ausprg}}{\ausfml} \limply \busfml}
}{}

\cinferenceRuleStore[US|US]{uniform substitution}
{\linferenceRule[formula]
  {\phi}
  {\applyusubst{\sigma}{\phi}}
}{}

\cinferenceRuleStore[linequs|$\exists$lin]{linear equation uniform substitution}
{\linferenceRule[impl]
  {b\neq0}
  {\big(\lexists{x}{(b\cdot x+c=0 \land q(x))}
  \lbisubjunct {q(-c/b)}\big)}
}{}

\dinferenceRuleStore[FA|FA]{First arrival}
{\ddiamond{\prepeat{\ausprg}}{\ausfml} \limply \ausfml \lor \ddiamond{\prepeat{\ausprg}}{(\lnot\ausfml\land\ddiamond{\ausprg}{\ausfml})}
}{}
\dinferenceRuleStore[Mor|M]{monotonicity axiom}
{\ddiamond{\ausprg}{(\ausfml\lor\busfml)}
\lbisubjunct
\ddiamond{\ausprg}{\ausfml} \lor \ddiamond{\ausprg}{\busfml}
}{}
\dinferenceRuleStore[VK|VK]{vacuous possible $\dbox{}{}$}
{\linferenceRule[impl]
  {p}
  {(\dbox{\ausprg}{\ltrue}{\limply}\dbox{\ausprg}{p})}
  \qquad
}{\m{\freevars{p}\cap \boundvars{\ausprg}=\emptyset}}
\dinferenceRuleStore[R|R]{Regular}
{\linferenceRule[formula]
  {\ausfml_1\land\ausfml_2\limply\busfml}
  {\dbox{\ausprg}{\ausfml_1} \land \dbox{\ausprg}{\ausfml_2} \limply \dbox{\ausprg}{\busfml}}
}{}

\title{Refactoring-as-Propositions: Proved Refactoring of Hybrid Systems via Proved Refinements}
\titlerunning{Refactoring-as-Propositions}

\author{Enguerrand Prebet\orcidID{0009-0008-0160-5219} \and
Andr\'e Platzer\orcidID{0000-0001-7238-5710}}

\authorrunning{E. Prebet, A. Platzer}

\institute{
  Karlsruhe Institute of Technology, Karlsruhe, Germany
  \email{\{enguerrand.prebet,platzer\}@kit.edu}
}

\maketitle              

\begin{abstract}
  Cyber-physical systems are inherently complex due to their connection between software and the physical world.
  Iterative design reduces their complexity, but increases the need to repeatedly recheck their safety in full after every change.
  We introduce the \emph{refactoring-as-propositions} principle in which refactorings are represented as propositions along with a method for proving that system refactorings preserve their required properties by transferring the proof along the respective modification.
  It is based on \emph{differential refinement logic} (\dRL), with which one can simultaneously and rigorously refer to properties of the systems and the relation between a refactored system and its original version.
  Refinements represent a uniform way of expressing different types of hybrid system refactorings, including those that introduce auxiliary variables.
  Furthermore, we show how these refactorings can be proved automatically, and/or reduce to a modular proof solely about the local change rather than about the whole system.
\end{abstract}

\section{Introduction}

Cyber-physical systems (CPS) are crucial but subtle to get right.
Their design needs to deal with several interacting features, making it challenging to get them all right all at the same time.
That is where iterative designs help when starting with simpler systems, figuring them out, and then iteratively adding more and more complicated features.
While this reduces the effort per design, the cost is that systems need to be analyzed for safety repeatedly, often with similar but different arguments.
CPS refactoring techniques \cite{DBLP:conf/fm/MitschQP14} offer specific refactoring operations on CPS that preserve safety when only certain differences are reproved.
But the correctness of the safety verification of the final system depends on the correct implementation and application of every refactoring operation.

This paper shows that all of those prior CPS refactoring operations can, instead, be derived syntactically as propositions in \emph{differential refinement logic} (\dRL) \cite{DBLP:conf/lics/LoosP16,DBLP:conf/ijcar/PrebetP24}.
Instead of verifying refactoring algorithms for specific properties, \emph{refactorings-as-propositions} are refinements proved in \dRL, and can be reused independently of the systems and properties considered.
The correctness of such refactoring applications then boils down to \dRL's refinements-to-properties axiom \irref{boxref}.
Since \dRL has a sound proof calculus \cite{DBLP:conf/lics/LoosP16}, syntactic derivations in \dRL give a sound implementation of CPS refactoring techniques.
Furthermore, \dRL's proof calculus is a uniform substitution calculus \cite{DBLP:conf/ijcar/PrebetP24}, so that the resulting derivations can be implemented in the prover \KeYmaeraX solely on top of its small soundness-critical prover microkernel.
This paves the way to provably correct automatic CPS refactoring technique implementations.

While most standard refactorings are usually meant to preserve the same semantics, it gets subtle for refactorings that introduce fresh auxiliary variables in imperative programs, which \dRL can model.
Consider adding the program $\ptest{x = 0};\pupdate{\umod{x}{1}}$ in a system where $x$ is fresh.
This snippet first tests if the value of $x$ is $0$ and aborts execution otherwise before setting the value of $x$ to $1$.
If neither the rest of the system nor the postcondition of interest use $x$, the corresponding execution should be unaffected, and so should be the validity of the postcondition.
Except this reasoning \emph{only} holds if the snippet is not added within a loop.
In a loop, the second iteration of the loop would abort since the test $x = 0$ would fail now that $x$ has been set to $1$.
This drastically alters the behavior despite not interacting with any existing variable!
Conversely, an assignment $\pupdate{\umod{x}{\astrm}}$ where $\astrm$ is any polynomial expression can safely be added no matter what variables occur in $\astrm$.
Having a formal proof for handling correct refactorings, based on the notion of \emph{ghost refinement}, gives a general understanding when such subtle changes are sound, while remaining automatable.

\paragraph{Contributions.}
We introduce the refactoring-as-propositions principle by providing techniques for transforming proved refinements in \dRL to generic, proved refactoring techniques for hybrid programs.
We do so for three classes of refinements: globally sound refinements, locally sound refinements (or conditional refinements), and ghost refinements, all of which require negligible changes to \dRL's axiomatization.
It allows for an implementation of these techniques in the theorem prover \KeYmaeraX \cite{DBLP:conf/cade/FultonMQVP15}, enabling sound automatic refactorings via refinement, while keeping the soundness-critical kernel small.
We show how these different classes of refactorings can be applied in case studies.

\paragraph{Related Works.}
Refactorings for differential dynamic logic (\dL) \cite{DBLP:conf/lics/Platzer12a}, which \dRL is an extension of, has been studied by adding a separate layer on top of the logic \cite{DBLP:conf/fm/MitschQP14}.
Other works include compositions \cite{DBLP:conf/acsd/LunelBT17} and parallelism \cite{DBLP:conf/cade/BriegerMP23}, which decompose a hybrid system into parts.
Thus, if a component is updated, then only part of the proof needs to be updated if it still satisfies the same contract.
The Event-B method \cite{DBLP:books/daglib/0024570} initially designed for discrete programs has been extended to continuous evolutions \cite{DBLP:conf/asm/BanachZSW12b,DBLP:conf/tase/DupontAPS19}.
The base logic of Event-B is not expressive enough for modelling continuous evolution so it must rely on theories \cite{DBLP:conf/birthday/ButlerM13}.
Hybrid Event-B \cite{DBLP:conf/asm/AfendiLM20,DBLP:conf/asm/AbrialSZ12} builds upon it to handle hybrid programs.
It can be used with the Rodin tool \cite{DBLP:journals/sttt/AbrialBHHMV10,DBLP:journals/scp/SuAZ14} to generate the proof obligations when proving refinements.
A detailed comparison between Event-B and \dL is provided in the literature \cite{DBLP:conf/asm/DupontAPS18}.

\section{Background: Differential Refinement Logic}
This section provides an overview of differential refinement logic (\dRL) \cite{DBLP:conf/lics/2016,DBLP:conf/ijcar/PrebetP24,DBLP:journals/scp/Platzer25}.
After recalling its syntax and semantics, we briefly present how the logic can be used for proofs, both in theory and using the theorem prover \KeYmaeraX \cite{DBLP:conf/cade/FultonMQVP15}.
\subsection{Syntax}
At its core, \dRL is a first-order multi-modal logic where the modalities are interpreted by hybrid programs, and with a dedicated refinement operator.
\dRL terms evaluate over real states, that is, valuations from variables to values in \R.
\begin{definition}[Term]
  \emph{Terms} are defined by the following grammar where $\astrm,\bstrm$ are terms, $x$ is variable, and $n$ is an integer.
\[
	\astrm, \bstrm \Coloneqq x \OR n
   \OR \astrm + \bstrm \OR \astrm \cdot \bstrm \OR \der{\astrm}
\]
\end{definition}
Beyond standard arithmetic, terms can also have differentials $\der{\astrm}$.
This is to account for the continuous fragment of hybrid programs.
For that purpose, each variable $x$ is associated a differential variable $\D{x}$.
Their behaviors are linked during differential equations, and differential variables are essential for the soundness of axioms like \irref{DIeq0} and \irref{refOde}.

\begin{definition}[Formula]
  \emph{Formulas} are defined by the following grammar where $\astrm,\bstrm$ are terms, $\asfml,\bsfml$ are formulas and $\asprg,\bsprg$ are hybrid programs (\rref{def:hp}).
\[
    \asfml, \bsfml \Coloneqq \astrm \leq \bstrm
    \OR \lnot \asfml \OR \asfml \land \bsfml \OR \lforall{x}{\asfml} \OR \dbox{\asprg}{\asfml} \OR \asprg \refines \bsprg
\]
\end{definition}
\dRL formulas may be true on only a subset of states.
We focus on the additional constructs not in first-order logic which are the modalities and refinements.
A formula $\dbox{\asprg}{\asfml}$ expresses that the formula $\asfml$ always holds after executing the hybrid program $\asprg$.
Thus the formula $\asfml \limply \dbox{\asprg}{\bsfml}$ can be understood as the standard Hoare triple $\{\asfml\} \asprg \{\bsfml\}$ but for hybrid systems.
By duality, there is another modality $\ddiamond{\asprg}{\asfml} \defeq \lnot \dbox{\asprg}{\lnot\asfml}$ that holds in a state $s$, if there exists an execution of $\asprg$ starting from $s$ whose final state satisfies $\asfml$.

Refinements $\asprg \refines \bsprg$ on the other hand relate the output states of $\asprg$, not to a formula $\asfml$ but to the output states of another program $\bsprg$.
A refinement $\asprg \refines \bsprg$ holds in a state $s$, if all states reachable by executing $\asprg$ from $s$ are also reachable by some execution of $\bsprg$ from $s$.
Program equivalence, written $\asprg \prgeq \bsprg$, simply denotes the symmetric version of refinement, and can be seen as syntactic sugar for $\asprg \refines \bsprg \land \bsprg \refines \asprg$.
Having refinements as a construct directly in the logic allows to reason and prove partial refinements, i.e., refinements that only holds under some assumption or after some execution like $\dbox{\asprg}{(\bsprg\refines\csprg)}$.
This idea is central for the refactoring techniques of \rref{sec:loc_ref}.

\begin{definition}[Hybrid Program]\label{def:hp}
  \emph{Hybrid program} are defined by the following grammar where $\astrm$ are terms, $\bsfml$ are formulas and $\asprg,\bsprg$ are hybrid programs.
\[
	\asprg, \bsprg \Coloneqq
  {\ptest{\bsfml}} \OR \pupdate{\pumod{x}{\astrm}} \OR \prandom{x} \OR \pode{\D{x}=\astrm}{\bsfml} \OR \pchoice{\asprg}{\bsprg} \OR \asprg;\bsprg \OR \prepeat{\asprg}
\]
\end{definition}

Hybrid programs form a Kleene Algebra with tests \cite{DBLP:journals/toplas/Kozen97} and represent relations from input states to output states.
The test $\ptest{\asfml}$ behaves as a noop in states where $\asfml$ holds and aborts otherwise.
The assignment $\pupdate{\umod{x}{\astrm}}$ updates the value of $x$ in the current state to the computed value of $\astrm$.
Similarly, the nondeterministic assignment $\prandom{x}$ updates the value of $x$ to any possible real value.
Practically, nondeterministic assignment are usually followed by a test to limit the possible values, e.g., $\prandom{x};\ptest{(0 \leq x \land x \leq 10)}$ updates $x$ to be any real inside the interval $[0,10]$.
Nondeterministic choice $\pchoice{\asprg}{\bsprg}$ can behave as $\asprg$ or as $\bsprg$.
Sequence composition $\asprg;\bsprg$ behaves as $\asprg$ first and then proceeds as $\bsprg$.
The \emph{differential equation} \(\pode{\D{x}=\astrm}{\bsfml}\) behaves like a continuous evolution where both the differential equation \(\D{x}=\astrm\) and the evolution domain $\bsfml$ holds and can be extended for multiple ODEs.
During the evolution, both the value of the variables $x$ and $\D{x}$ are modified.
The duration of the evolution is nondeterministic:
it may stop immediately, in which case only $\D{x}$ is modified, or continue up until $\bsfml$ is about to be violated, etc.
The evolution domain $\bsfml$ can be used to model physical reality -- e.g., a ball stays above the floor -- or internal triggers by stopping the evolution so that another part of the system is executed before continuing.
Finally, the Kleene star $\prepeat{\asprg}$ executes $\asprg$ a nondeterministic number of times, possibly none.

The highly nondeterministic behavior of hybrid programs pairs well with the box modality as it ensures the postcondition for all possible executions.
More concrete, deterministic, programs can then be proved safe by refining the nondeterministic one.
Another benefit of nondeterminism is regarding differential equations.
By definition $\dbox{\pode{\D{x}=\astrm}{\bsfml}}{\asfml}$ means that $\asfml$ only holds in the final state.
However, since the duration can be arbitrary, proving such formula amounts to proving that $\asfml$ holds throughout the whole continuous evolution.

\subsection{Calculus}

\dRL comes with an axiomatization that is used to prove validity of its formulas.
Given that formulas may hold only on subsets of states, we say that a formula is \emph{valid} if it holds for all states.
Similarly, a proof rule is sound if the validity of its premises implies the validity of its conclusion.

In \rref{fig:dRL}, we present a subset of axioms and rules of \dRL \cite{DBLP:conf/ijcar/PrebetP24} to highlight its expressiveness and reasoning capabilities.
\begin{figure}
\begin{calculuscollection}
\begin{calculus}
  
\cinferenceRule[assignb|$\dibox{:=}$]{assignment / substitution axiom}
{\linferenceRule[equiv]
  {\asfml(\astrm)}
  {\axkey{\dbox{\pupdate{\umod{x}{\astrm}}}{\asfml(x)}}}
}
{}

\dinferenceRule[DIeq0|DI]{differential invariant axiom}
{\linferenceRule[lpmi]
  {\big(\axkey{\dbox{\pevolve{\D{x}=\astrm}}{\,\bstrm=0}} \lbisubjunct \bstrm=0\big)}
  {\dbox{\pevolve{\D{x}=\astrm}}{\,\axeffect{\der{\bstrm}=0}}}
}
{}

\cinferenceRule[boxref|{[$\refines$]}]{Box refinement}
{\linferenceRule[impl]
  {\asprg \refines \bsprg}
  {(\dbox{\bsprg}{\asfml}\limply \axkey{\dbox{\asprg}\asfml})}
}{}

\cinferenceRule[leqtrans|$\refines_{t}$]{Transitivity of $\refines$}
{\linferenceRule[impl]
  {\asprg \refines \bsprg}
  {(\bsprg \refines \csprg \limply \axkey{\asprg \refines \csprg})}
}{}

\cinferenceRule[sequence|$;_{\refines}$]{sequence axiom}
{\linferenceRule[lpmi]
{\axkey{\asprg;\bsprg \refines \csprg;\dsprg}}
{\asprg \refines \csprg \land \dbox{\asprg}{\bsprg \refines \dsprg}}
}{}

\cinferenceRule[refOde|ODE]{ODE axiom}
{\linferenceRule[equiv]
  {\dbox{\pode{\D{x}=\astrm}{\asfml}}{(\D{x}=\bstrm \land \bsfml)}}
  {(\axkey{\pode{\D{x}=\astrm}{\asfml}\refines\pode{\D{x}=\bstrm}{\bsfml}})}
}{}
\end{calculus}
\begin{calculus}
\cinferenceRule[G|G]{$\dbox{}{}$ generalization} 
{\linferenceRule[formula]
  {\asfml}
  {\dbox{\asprg}{\asfml}}
}{}%

\cinferenceRule[composeb|$\dibox{{;}}$]{composition} 
{\linferenceRule[equiv]
  {\dbox{\asprg}{\dbox{\bsprg}{\asfml}}}
  {\axkey{\dbox{\ausprg;\bsprg}{\asfml}}}
}{}%

\cinferenceRule[leqrefl|$\refines_{r}$]{Reflexivity of $\refines$}
{\linferenceRule[leq]
  {\asprg}
  {\asprg}
}{}

\cinferenceRule[test|?]{test axiom}
{\linferenceRule[equiv]
{(\asfml \limply \bsfml)}
{(\axkey{\ptest{\asfml} \refines  {\ptest{\bsfml}}})}
}{}
\end{calculus}
\end{calculuscollection}
\caption{Some \dRL axioms and rules}
\label{fig:dRL}
\end{figure}
If the differential of a term is 0 during continuous evolution, it is a differential invariant, and the value of the term is determined by its initial value (\irref{DIeq0}).
Axioms \irref{composeb} and \irref{sequence} decompose a proof for a modality and refinement into a proof for their respective constituant.
Axiom \irref{boxref} says that if all outputs of a program $\bsprg$ satisfy a formula $\asfml$, then all outputs of any refined program $\asprg$ also do.
Refinement is reflexive (\irref{leqrefl}) and transitive (\irref{leqtrans}).
Axiom \irref{test} relates a refinement between tests and their respective formula.
A refinement between differential equations only holds if the equation of the more general holds throughout the execution of the more refined program (\irref{refOde}).
When the domain does not change, this gives a program equivalence instead.
Additional axioms and rules used but not central to this paper are given in \rref{app:axs}.
The basic usage of the axioms and rules is to use propositional reasoning to reduce the expression in \axkey{blue} by a simpler one, where the hybrid programs are decomposed in simpler subprograms.

\paragraph{\KeYmaeraX.}

The theorem prover \KeYmaeraX \cite{DBLP:conf/cade/FultonMQVP15} implements \dL, and more recently its microkernel now also supports its extension \dRL\cite{DBLP:conf/ijcar/PrebetP24}.
It uses a uniform substitution calculus version\cite{DBLP:conf/cade/Platzer19,DBLP:conf/ijcar/PrebetP24} of the one presented here.
Axioms are expressed directly in the logic, with a unique rule -- the uniform substitution rule -- recovering all its sound instantiations.
Uniform substitution is the key to \KeYmaeraX's small soundness-critical microkernel.
Basing the refactoring techniques on the \dRL proofs thus provides more automation, while limiting the changes done to the kernel to the new axioms required: \rref{lem:cprge} and \rref{lem:ghost_ode}.

\section{Three Levels of Refactoring}
The objective of a CPS refactoring is to change a part of a program without affecting the safety of the program or its resulting proof.
The reason such change is valid can be justified on three different levels, from a local argument that can be proved automatically to a more general reasoning that requires fixing the proof to accommodate the change. The proofs from this section are accessible in \rref{app:312} and \rref{app:33}.

\subsection{Global Refinement}
\label{sec:glob_ref}

A refinement between two systems always ensures that a proof using one correctly transfers to a proof of the other.
It is highlighted by the axioms \irref{boxref} and \irref{leqtrans}.
In both cases, a proof for a refinement $\asprg \refines \bsprg$ ensures that the formula in \axkey{blue} can be proved by proving the same formula but with the program $\bsprg$ instead.
Thus, any refactoring from $\bsprg$ to $\asprg$ immediately preserves all the corresponding properties, whether about safety or refinement.
Note that the axiom \irref{leqtrans} is equivalent -- up-to renaming and propositional reasoning -- to \(\bsprg \refines \asprg \limply (\csprg \refines \bsprg \limply \axkey{\csprg \refines \asprg})\) which allows reasoning on the right-hand side too, but requires a proof of the converse refinement: \(\bsprg \refines \asprg\).
\dRL axiomatization contains all axioms -- and thus all facts -- about Kleene Algebra with tests, including for instance:

\begin{calculuscollection}
\begin{calculus}
\cinferenceRule[choiceL|$\cup_l$]{Refinement choice left}
{\linferenceRule[leq]
  {\pchoice{\asprg}{\bsprg}}
  {\asprg}
}{}
\cinferenceRule[seqidl|$;_{\text{id-l}}$]{Left neutral of $;$}
{\linferenceRule[eq]
  {\asprg}
  {\axkey{\ptest{\ltrue};\asprg}}
}{}
\end{calculus}
\qquad
\begin{calculus}
    \cinferenceRuleQuoteDef{seqdistl}
    \cinferenceRuleQuoteDef{seqdistr}
\end{calculus}
\end{calculuscollection}

\paragraph{Congruence.}
On their own, the axioms \irref{leqtrans} and \irref{boxref} are limited to top-level refinements about the whole program.
This can be corrected to extend the applicability of global refinement.
If the refinement is not on the full program but only a subcomponent, this principle can be extended via general congruence rules.
\rref{fig:glob_ref} presents one such refactoring extending \irref{boxref} for subcomponents.
\irlabel{Ax|Ax}
\begin{figure}
\begin{sequentdeduction}[array]
\linfer[boxref]{
  \lsequent{\Gamma}{\dbox{C(\bsprg)}{\asfml}}
  ! \linfer{
    \linfer{
      \linfer[Ax]{\lclose}
      {\lsequent{}{\asprg \refines \bsprg}}
    }{&\quad\vdots \text{ ind. on }C}
  }{\lsequent{\Gamma}{C(\asprg) \refines C(\bsprg)}}
}{\lsequent{\Gamma}{\dbox{C(\asprg)}{\asfml}}}
\end{sequentdeduction}
\caption{Refactor template for box in context via global refinement}
\label{fig:glob_ref}
\end{figure}

Its soundness revolves around the notion of monotone (resp.\ antitone) contexts.
Contexts are expressions (formulas or programs) with a hole.
The hole $\cdot_p$ (resp.\ $\cdot_f$) expects a program (resp.\ a formula), and we denote by $C(e)$ the expression obtained by substituting the hole by the expression $e$ assuming the kind are respected.
Monotone (resp.\ antitone) contexts are specific contexts that preserve (resp.\ commute) implications and refinements when applied.
The polarity of contexts expresses whether they are monotone ($+$) or antitone ($-$).
\begin{definition}[Monotone context]\label{def:mono_ctx}
  We define simultaneously monotone and antitone contexts where both the hole and the resulting expression can be formulas or programs as $C_{k_1 k_2}^\pm$, where $\pm \in \{+,-\}$ denotes the polarity, $k_1 \in \{f,p\}$ the kind of the resulting expression, and $k_2 \in \{f,p\}$ the kind of the hole.
  \begin{align*}
    C_{kk}^+ & \Coloneqq \cdot_k\\
    C_{pk}^\pm & \Coloneqq \asprg \OR {\ptest{C_{fk}^\pm}} \OR \pode{\D{x}=\astrm}{C_{fk}^\pm} \OR C_{pk}^\pm;C_{pk}^\pm \OR \pchoice{C_{pk}^\pm}{C_{pk}^\pm} \OR \prepeat{C_{pk}^\pm}\\
    C_{fk}^\pm & \Coloneqq \asfml \OR C_{fk}^\pm \land C_{fk}^\pm \OR \lnot C_{fk}^\mp \OR \lforall{x}C_{fk}^\pm \OR \dbox{C_{pk}^\mp}{C_{fk}^\pm} \OR C_{pk}^\mp \refines C_{pk}^\pm
  \end{align*}
\end{definition}
The kind $k \in \{f,p\}$ specifies the kind of the hole (formula or program), e.g., the hole $\cdot_p$ is a program to program context, i.e., $C_{pp}^+$.
$\mp$ stands for the polarity opposite to $\pm$, e.g., for antitone program to formula contexts the grammatical rule for box modalities is $C_{fp}^- \Coloneqq \dbox{C_{pp}^+}{C_{fp}^-}$.
The only constructs switching polarities are the negation, the box modality (for the program), and the refinement (for the left-hand side program).
It is implicitly extended to derived constructs, e.g., \(C_{fk}^\mp \limply C_{fk}^\pm\) as \(\asfml \limply \bsfml \equiv \lnot (\asfml \land \lnot \bsfml)\).
Note that these contexts may not have a hole under an equivalence, e.g., \(\cdot_f \lbisubjunct \bsfml\) would be \((\cdot_f \limply \bsfml) \land (\bsfml \limply \cdot_f)\) which is not allowed by the grammar, as both implications have conflicting polarities.

\begin{theorem}[Monotone congruence]\label{thm:useAt}
  Given contexts $C_{pp}^+$, $C_{fp}^-$, $C_{pf}^-$ (and similarly for $C_{pp}^-,C_{fp}^+, C_{pf}^+, C_{ff}^\pm$), then the following rules are derivable in \dRL{}:
  \vspace*{-0.1cm}
  \begin{mathpar}
    \linfer{\asprg \refines \bsprg}{C_{pp}^+(\asprg) \refines C_{pp}^+(\bsprg)}
    \and
    \linfer{\bsprg \refines \asprg}{C_{fp}^-(\asprg) \limply C_{fp}^-(\bsprg)}
    \and
    \linfer{\bsfml \limply \asfml}{C_{pf}^-(\asfml) \refines C_{pf}^-(\bsfml)}
  \end{mathpar}
\end{theorem}
The theorem is proved by induction on the contexts, using axioms like \irref{sequence} to reduce each construct to its subcomponents.
The template of \rref{fig:glob_ref} is now a simple application of \rref{thm:useAt} for $C_{pp}^+$ contexts.
Note that even for $C_{ff}$ contexts, the proof may require the use of refinement, for instance with \(\dbox{\ptest{\cdot_f}}\asfml\) since the hole is in program (itself in a formula).
Given the mutually inductive nature of formulas and programs in \dRL, having a program counterpart of implications leverages this kind of induction for obtaining much more general principles.

\paragraph{Dealing with equivalence.}
We denote by $C_{k_1 k_2}$ with $k_1,k_2\in\{f,p\}$ general contexts, defined by the same grammar as \rref{def:mono_ctx} by dropping all polarities.
It only ensures the holes are all of the same kind.
Since polarities are dropped, \(\cdot_f \lbisubjunct \bsfml\) is a valid $C_{ff}$ context.
Compared to refinements, given a global program equivalence, i.e., a proof $\lsequent{}{\asprg \prgeq \bsprg}$, it is always sound to replace the program $\asprg$ by $\bsprg$ in a context, without considering monotonicity or free and bound variables.
Program equivalences are dealt with via a specific rule \irref{CPrgE}, similar to the \dL congruence rules for terms and formulas \cite{DBLP:journals/jar/Platzer17}.

\begin{lemma}[Contextual program equivalence]\label{lem:cprge}
  Rule \irref{CPrgE} is admissible.
\begin{center}
\cinferenceRule[CPrgE|\text{CPE}]{congequiv congruence of equivalences on formulas}
{\linferenceRule[formula]
  {\asprg \prgeq \bsprg}
  {\contextapp{C_{fp}}{\asprg} \lbisubjunct \contextapp{C_{fp}}{\bsprg}}
}{}%
\end{center}
\end{lemma}
Rule \irref{CPrgE} circumvents the need of an induction on $C_{fp}$, and so the corresponding template for program equivalence, given in \rref{fig:glob_eq}, is simpler than for refinement.

\begin{figure}
\begin{sequentdeduction}[array]
\linfer[MPax]{
  \lsequent{\Gamma}{\dbox{C(\bsprg)}{\asfml}}
  ! \linfer[equivifyr]{\linfer[CPrgE]{
      \linfer[Ax]{\lclose}
      {\lsequent{}{\bsprg \prgeq \asprg}}
  }{\lsequent{\Gamma}{\dbox{C(\bsprg)}{\asfml} \lbisubjunct \dbox{C(\asprg)}{\asfml}}}}
  {\lsequent{\Gamma}{\dbox{C(\bsprg)}{\asfml} \limply \dbox{C(\asprg)}{\asfml}}}
}{\lsequent{\Gamma}{\dbox{C(\asprg)}{\asfml}}}
\end{sequentdeduction}
\caption{Refactor template via global program equivalence}
\label{fig:glob_eq}
\end{figure}

Having both these templates transforms simple \dRL axioms into refactoring techniques, allowing to freely rewrite hybrid programs during the proofs.
For instance, with \irref{seqdistl}, \irref{seqdistr}, \irref{refOde}, and \irref{choiceL}, we immediately obtain formal generic versions of the refactorings (R1--2), (R6--7), and (R9--10) respectively \cite{DBLP:conf/fm/MitschQP14}, which are given in \rref{fig:glob_refact}.
Since (R6--7) and (R9--10) are obtained from refinements and not program equivalences, they are only allowed in one direction.
Box modality results transfer via (R6) and (R10), while diamond modality results transfer via (R7) and (R9).
\begin{figure}
  \centering
  \begin{tikzpicture}[decoration=zigzag]
      \node (a1) at (0,0) {$\bigcup_i(\asprg_i;\bsprg)$};
      \node (a2) at (3,0) {$(\bigcup_i\asprg_i);\bsprg$};
      \draw[->] (a1.north east) to [out=30,in=150] node[midway, above] {(R1)} (a2.north west) ;
      \draw[<-] (a1.south east) to [out=-30,in=-150] node[midway, below] {(R2)} (a2.south west);

      \node (b1) at (0+7,0) {$\pode{\D{x}=\astrm}{\asfml}$};
      \node (b2) at (3+7,0) {$\pode{\D{x}=\astrm}{\asfml \land \bsfml}$};
      \draw[->] (b1.north east) to [out=30,in=150] node[midway, above] {(R6)} (b2.north west) ;
      \draw[<-] (b1.south east) to [out=-30,in=-150] node[midway, below] {(R7)} (b2.south west);

      \node (c1) at (0+3.5,-2) {$\asprg;\bsprg$};
      \node (c2) at (3+3.5,-2) {$(\pchoice{\asprg}{\csprg});\bsprg$};
      \draw[->] (c1.north east) to [out=30,in=150] node[midway, above] {(R9)} (c2.north west) ;
      \draw[<-] (c1.south east) to [out=-30,in=-150] node[midway, below] {(R10)} (c2.south west);

  \end{tikzpicture}
  \caption{Program refactorings induced by global refinements}
  \label{fig:glob_refact}
\end{figure}
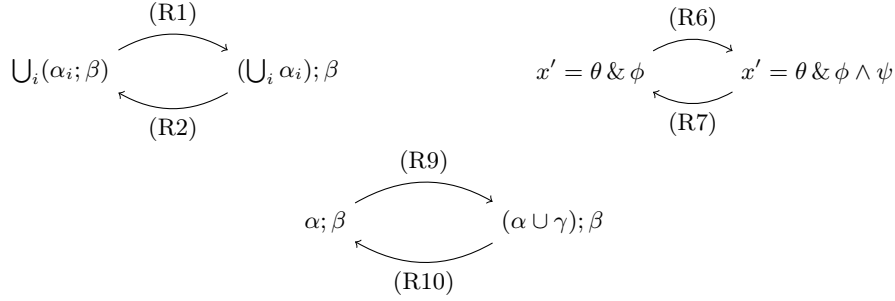
\vspace*{-1.5em}
\subsection{Local Refinement}
\label{sec:loc_ref}
Sometimes, local changes are only justified due to the surrounding context.
For instance, a test may evaluate to false due to a previous assignment, making one side of a choice $\pchoice{\asprg}{\csprg}$ unreachable.
In that case, the other direction for the refactorings (R9--10) of \rref{fig:glob_refact} is correct.
To allow these refactorings to occur, we must be able to reason about the local information that is known at the position where the programs lie.
This is an issue for the template of \rref{fig:glob_ref} since all context is lost in the congruence, and is answered by considering \emph{local refinements}.

Whereas global refinements enable direct usage of the axioms of \dRL, local refinements are useful for applying more complex lemmas, where the refinement or implication holds thanks to some side condition.
Compared to the examples of \rref{sec:glob_ref}, the refinement statements are of the form $\bsfml \limply \asprg \refines \bsprg$ or $\dbox{\csprg}{\asprg \refines \bsprg}$.
This makes it a more powerful but also more manual refactoring.
The conditional refinement may need to be proved by the user beforehand, and the assumption $\bsfml$ has to hold locally, which must also be proved.
Nonetheless, using the proved lemma for refactoring can still be done automatically, similarly to the congruence implementation for \rref{sec:glob_ref}.
The proof is altered to keep the original assumptions $\Gamma$, so as to prove that the refinement holds at the position where it is used.
From here and onward, we will focus on single-hole program contexts ($C_{pp}^+$) which will be noted $C$ when unambiguous.

\paragraph{Towards context-aware congruence.}

To preserve the context, the idea is to look for a congruence property that holds in the same state.
Implication-based and equivalence-based axioms are better suited for this task than proof rules as the conclusion must hold in the same state as the assumption.
This makes it possible to use them even if the refinement does not hold globally.
In that regards, consider the right subcomponents in \irref{sequence}.
The refinement we obtain ($\bsprg \refines \dsprg$) is now under a modality to express the fact that it \emph{only} needs to hold in the states reachable from the current state after any possible execution of the context program $\asprg$.
Given a single-hole context $C$, we can construct the program that must be executed before the hole using \emph{\projctx{}s}.
\begin{definition}[\Projctx{}]
Given a monotone program context $C$ and a program $\asprg$, we define the \projctx{} $C_\asprg$ as a formula-to-formula context inductively on the structure of $C$:
\begin{itemize}
  \item when $C \equiv \cdot_p$, the context $C_\asprg$ is $\cdot_f$
  \item when $C \equiv \csprg;C'$, the context $C_\asprg$ is $\dbox{\csprg}{C'_\asprg}$
  \item when $C \equiv \prepeat{C'}$, the context $C_\asprg$ is $\dbox{\contextapp{C}{\asprg}}{C'_\asprg}$
  \item when $C \equiv C';\csprg$, or $\pchoice{C'}{\csprg}$, or $\pchoice{\csprg}{C'}$, the context $C_\asprg$ is $C'_\asprg$
\end{itemize}
\end{definition}
The dependency on $\asprg$ only shows in loops.
\rref{lem:focus} would still hold when using $C_\bsprg$, though it would be less practical in general.
Indeed, even though both rules are derivable, the premise with $C_\bsprg$ may not hold even when ours is.

\begin{lemma}\label{lem:focus}
  For a monotone program context $C$, the following rule is derivable in \dRL{}:
  \[
    \linfer{\lsequent{\Gamma}{C_\asprg(\asprg \refines \bsprg)}}{\lsequent{\Gamma}{\contextapp{C}{\asprg} \refines \contextapp{C}{\bsprg}}}
  \]
\end{lemma}
Using \rref{lem:focus}, we obtain a refactoring whose structure is similar to \rref{fig:glob_ref} and shown in \rref{fig:loc_ref}, except that the right branch now keeps the assumptions $\Gamma$ by using the \projctx{} $C_\asprg$.
Closing the goal $\lsequent{\Gamma}{C_\asprg(\asprg \refines \bsprg)}$ then depends on the form of the refinement statement used.
For statements like $\bsfml \limply \asprg \refines \bsprg$, it reduces to $\lsequent{\Gamma}{C_\asprg(\bsfml)}$ by monotonicity, which still need to be proved.
For statements like $\dbox{\csprg}{\asprg \refines \bsprg}$, it can be concluded immediately when the last modality of $C_\asprg$ is $\dbox{\csprg}$, by using \irref{G} to remove the outer modalities.

\vspace*{-0.4em}
\begin{figure}
\begin{sequentdeduction}[array]
\linfer[boxref]{
  \lsequent{\Gamma}{\dbox{C(\bsprg)}{\asfml}}
  ! \linfer{
    \linfer{
      \lsequent{\Gamma}{C_\asprg(\asprg \refines \bsprg)}
    }{&\vdots \text{ ind. on $C$ (\rref{lem:focus})}}
  }{\lsequent{\Gamma}{C(\asprg) \refines C(\bsprg)}}
}{\lsequent{\Gamma}{\dbox{C(\asprg)}{\asfml}}}
\end{sequentdeduction}
\caption{Refactor template via local refinement}
\label{fig:loc_ref}
\end{figure}
\vspace*{-1.3em}
\paragraph{Program Equivalence.}
In order to preserve local information, program equivalences need to be treated in the same way as local refinements.
As the \projctx{} $C_\asprg$ is dependent on the structure of $C$, the congruence cannot be implemented in just one rule like with \rref{sec:glob_ref}.
\begin{corollary}\label{cor:focusEq}
  For a program context $C$, the following rule is derivable in \dRL{}:
  \[
    \linfer{\lsequent{\Gamma}{\contextapp{C_\asprg}{\asprg \prgeq \bsprg}}}{\lsequent{\Gamma}{C(\asprg) \prgeq C(\bsprg)}}
  \]
\end{corollary}

Refactorings via local refinements are essential as it can be relatively frequent that equivalences that arise from cyber-physical systems rely on generic assumptions like non-negativity of some parameters.
Even if generally assumed, the presence of these small assumptions makes the argument inherently local.
\rref{sec:acas} shows how that would apply in a realistic setting.

\subsection{Ghost Equivalence}
\label{sec:ghost}
The last type of refinements involve introducing new variables.
Inside modalities, it is simple to add fresh variables: for instance by \irref{assignb}, $\dbox{\pupdate{\pumod{x}{\astrm}}}{\asfml} \lbisubjunct \asfml$ holds whenever $x$ is fresh, and introduces the fresh variable $x$ as a ghost variable.
Comparatively, as refinements look at the output value of all variables, its meaning is strongly affected by the addition of new variables.
However, it is still possible in \dRL to express refinements that disregard the value of a variable: the refinement $\asprg \refines \bsprg;\prandom{x}$ holds when the program $\asprg$ refines $\bsprg$ without considering the final value of $x$.
For program equivalences, the corresponding version is symmetric: $(\asprg;\prandom{x}) \prgeq (\bsprg;\prandom{x})$, which we will focus on for this section.
This technique enables additional refactoring principles:

\vspace{0.2em}
\cinferenceRule[odeconst|{ODE$_{cst}$}]{ODE const}
{\linferenceRule[eq]
  {(\pupdate{\umod{x}{\astrm}};\pode{\D{y}=x}{\bsfml};\prandom{x})}
  {\axkey{\pode{\D{y}=\astrm}{\bsfml}};\prandom{x}}
}{$x,y\notin\astrm; x\notin\bsfml$}
\\[-0.8em]

\cinferenceRule[refDG|{$\prgeq_{dG}$}]{Ref dG}
{\linferenceRule[eq]
  {(\pupdate{\umod{x}{c}};\pode{\D{y}=\astrm,\D{x}=\bstrm}{\bsfml};\prandom{x};\prandom{\D{x}})}
  {\axkey{\pode{\D{y}=\astrm}{\bsfml}};\prandom{x};\prandom{\D{x}}}
}{$\bstrm$ linear in $x$, $x\notin\astrm,\bsfml,y$}
\vspace{0.2em}

The latter program equivalence is the \dRL counterpart of \dL axiom dG\cite{DBLP:conf/lics/PlatzerT18}.

\begin{lemma}\label{lem:ghost_ode}
The axioms \irref{odeconst} and \irref{refDG} are sound.
\end{lemma}

These new axioms are focused on introducing ghost variables near differential equations.
Similar equivalences for assignments can also be used, as long as they are provable in \dRL, e.g., \(\pupdate{\umod{x}{\astrm}};\prandom{x} \prgeq \prandom{x}\).

\paragraph{Congruence.}
The drawback of the version using refinement is the presence of these nondeterministic assignments.
To perform refactorings on a subprogram, these assignments need to be moved in position to ensure that the full program will not be affected by the introduction of the new variable.

\rref{fig:ghost_ref} showcases the structure to apply a ghost program equivalence of the form ${\asprg;\prandom{x} \prgeq \bsprg;\prandom{x}}$.
First, we need to introduce the nondeterministic assignment of $x$, which only succeeds if $x$ is effectively fresh in $\asfml$.
Then, it is moved next to $\asprg$.
Though the local reasoning is different from the previous refactoring, the general structure is similar and done by deriving the proof by induction on the structure of $C$.
Compared to the previous inductions, the case for the loop is particularly tricky and discussed in more detail below.
\begin{theorem}
  \label{thm:fresh}
If $x$ is fresh in $C$ and $x$ is not free in $\asprg$, then the following program equivalence is derivable in \dRL{}:
\[
C(\asprg);\prandom{x} \prgeq C(\asprg;\prandom{x});\prandom{x}
\]
\end{theorem}
Note that it would be unsound to remove the final nondeterministic assignment in general.
Indeed, $\prepeat{(\pupdate{\umod{c}{c+1}};\prandom{x})};\prandom{x}$ and $\prepeat{(\pupdate{\umod{c}{c+1}};\prandom{x})}$ are not equivalent since only the former program can modify $x$ but not $c$.
The program equivalence rule \irref{CPrgE} enables the replacement of $\asprg$ by $\bsprg$ using the congruence for global refinement from \rref{sec:glob_ref}.
The rest of the refactoring of \rref{fig:ghost_ref} uses \rref{thm:fresh} again to removing the leftover assignment.
\vspace*{-0.5em}
\begin{figure}
  \renewcommand{\linferPremissSeparation}{\hspace{0.1cm}}
\resizebox{\textwidth}{!}{
\begin{minipage}{1.2\textwidth}
\begin{sequentdeduction}[array]
\linfer[assignb]{
  \linfer[boxref]{
    \linfer[boxref]{
      \linfer{
        \linfer{
          \linfer[assignb]{
            \lsequent{\Gamma}{\dbox{C(\bsprg)}{\asfml}}
          }{
            \lsequent{\Gamma}{\dbox{C(\bsprg);\prandom{x}}{\asfml}}
          }
        }{&\quad\vdots \text{ (\rref{thm:fresh})}}
      }{\lsequent{\Gamma}{\dbox{C(\bsprg;\prandom{x});\prandom{x}}{\asfml}}}
      ! \linfer[CPrgE]{
          \linfer[Ax]{\lclose}{\lsequent{}{\asprg;\prandom{x} \prgeq \bsprg;\prandom{x}}}
          }{&\quad\dots}
    }{\lsequent{\Gamma}{\dbox{C(\asprg;\prandom{x});\prandom{x}}{\asfml}}}
    ! \linfer[leqantisym]{
      \linfer{
        \linfer{\lclose}{&\quad\vdots \text{ ind. on $C$ (\rref{thm:fresh})}}
      }{
      \lsequent{\Gamma}{C(\asprg);\prandom{x} \prgeq C(\asprg;\prandom{x});\prandom{x}}  
      }
    }{\lsequent{\Gamma}{C(\asprg);\prandom{x} \refines C(\asprg;\prandom{x});\prandom{x}}}
  }{\lsequent{\Gamma}{\dbox{C(\asprg);\prandom{x}}{\asfml}}}
}{\lsequent{\Gamma}{\dbox{C(\asprg)}{\}asfml}}}
\end{sequentdeduction}
\end{minipage}
}
\caption{Refactor template via ghost refinement}
\label{fig:ghost_ref}
\end{figure}
\vspace*{-1.1em}
\paragraph{Handling loops.}
Looking at the induction case for loops, we want to automatically derive $\prepeat{(C(\asprg);\prandom{x})};\prandom{x} \prgeq \prepeat{C(\asprg)};\prandom{x}$.
Due to the repetition, even though the assignment is after the program, it may influence the execution before the next iteration.
Thus, the proof of the equivalence above relies on $(\prandom{x};C(\asprg);\prandom{x}) \prgeq C(\asprg);\prandom{x}$, which only holds if $x$ is not free \cite[Def.15]{DBLP:journals/jar/Platzer17} in $C(\asprg)$, and is proved by induction on $C$.
Using the template of \rref{fig:ghost_ref} with \irref{odeconst} or \irref{refDG} requires applying \rref{thm:fresh} on a program that mentions the introduced variable $x$.
To remain applicable, \rref{thm:fresh} cannot assume that $x$ is fresh for $\asprg$ like it is for $C$.
This weaker condition holds for both refinements considered, \irref{odeconst} and \irref{refDG}, as the right program always start with an assignment on $x$.
This unexpected other induction strengthens the need for having such reasoning automated, as the proof would grow very fast.
\section{Applications}
This section demonstrates how these refactorings via refinements are applied at various levels of abstraction, from the very concrete implementation in \KeYmaeraX, to a case study showcasing the proof simplification provided by proved refactorings, to a more abstract method for proving time-triggered models out of event-triggered models.\footnote{\url{https://github.com/EnguerrandPrebet/Refactoring-as-Propositions}}

\subsection{Implementation in \KeYmaeraX}
The source code containing the extension of \KeYmaeraX for \dRL is available online\cite{prebet_2026_20142337}.
The implementation consists of two parts.
First, the new axioms and rules need to be added to the soundness-critical kernel.
This includes \rref{lem:cprge}, \rref{lem:ghost_ode}.
Additionally, to prevent unsound renaming, the kernel disallows instantiating by a differential variable axioms that are phrased for plain variables (such as in \irref{assignb}).
This showed a completeness issue in the previous implementation of \dRL.
As the axiom \irref{refDG} deals with both kinds of variables, manipulating these nondeterministic assignments also required to duplicate some preexisting axioms with differential variants.

Secondly, the refactorings are implemented as tactics, i.e., automated proofs that combine \verb|provable|s, \KeYmaeraX's certificate of sound inference rules.
This tactics are then made accessible to the user via the tactic language \cite{DBLP:conf/itp/FultonMBP17}.
This provides most of the automation to build complex proofs without being soundness-critical.
The global refinements from \rref{sec:glob_ref} led to the extension of the \verb|useAt| tactic which handles congruence for implications and formula equivalences.
Without refinement, congruence proofs for implication were limited to $C_{ff}^\pm$ contexts where the hole is not under a program, e.g., \(\dbox{\ptest{\cdot_f}}\asfml\) was not supported.
The tactic now supports all monotone and antitone contexts from \rref{def:mono_ctx}, by deriving the approriate sound rule from \rref{thm:useAt}.
Refactorings like \rref{fig:glob_ref} are executed simply by providing the \verb|useAt| tactic with the refinement axiom used as well as the position of $\asprg$ in the sequent.
Local refinements are given a specific tactic \verb|focus| which \emph{automatically} computes the \projctx{} and then derives the sound rule of \rref{lem:focus}, focusing a refinement down to its core change.
The extension to \rref{cor:focusEq} is similar, but since it is not required for the application below, it is left as future work.
Finally, the proof of \rref{thm:fresh} is automatically constructed using the new axiom \irref{refRandomComm}.
The interface takes the form of two tactics, \verb|moveRandomIn| and \verb|moveRandomOut|, that uses \rref{thm:fresh} to rewrite the goal in one way or another.
Thus, \rref{fig:ghost_ref} is done by using \irref{assignb} to introduce the fresh variable, \verb|moveRandomIn| to copy it next to $\asprg$, then applying the ghost equivalence (via \verb|useAt|), before removing the nondeterministic assignment by \verb|moveRandomOut| and \irref{assignb}.

\subsection{Next-generation Airborne Collision Avoidance System ACAS~X}
\label{sec:acas}

ACAS~X is a Airborne Collision Avoidance System developed by the Federal Aviation Administration (FAA).
It gives vertical advisories to a pilot to prevent the plane from colliding with other aircraft.
Previous formal verification of ACAS~X \cite{DBLP:conf/tacas/JeanninGKGSZP15} provides hybrid models safe under various assumptions (e.g., pilot's reaction time).
Crucially, all these models rely on the definition of the region where an advisory is safe.
We focus on the first model where a formula $L^{-1}$ expresses that a given advisory will always be safe, i.e., no further advisory will be required.
Two formulas expressing that property where given:
\begin{itemize}
  \item An implicit domain formulation $L^{-1}_\text{impl}$, that quantifies over multiple variables including time to ensure that the computed positions of the aircraft are sufficiently far apart.
  This formula's intent is clear which makes the proof simpler, but the presence of quantifiers makes the effective checking of such formula inefficient and thus not suited for implementation in a real aircraft.
  \item An explicit domain formulation $L^{-1}_\text{expl}$, where the space is split in different modes, leading to a quantifier-free formula, that can be computed quickly in an actual controller.
\end{itemize}

The ACAS~X safety theorem considered has the form
\[
(\mathit{init} \land L^{-1}) \limply \dbox{\prepeat{((\pchoice{\mathit{advisory};\ptest{L^{-1}}}{\mathit{skip}});\mathit{acc};\mathit{motion})}}{\mathit{nocol}}
\]
At each iteration, if not skipped, a new advisory satisfying $L^{-1}$ is chosen, then the acceleration is updated according to the advisory before the two aircraft move in a continuous $\mathit{motion}$.
Assuming a safe start ($L^{-1}$) and some initial condition, the system always ensures the absence of collision $\mathit{nocol}$.
This formula was proved using $L^{-1}_\text{impl}$ \cite[Theorem 1]{DBLP:conf/tacas/JeanninGKGSZP15}.
These two formulations are then proved equivalent assuming $\mathit{init}$, i.e., \(\mathit{init} \limply (L^{-1}_\text{impl} \lbisubjunct L^{-1}_\text{expl})\) \cite[Lemma 1]{DBLP:conf/tacas/JeanninGKGSZP15}.
Because that equivalence is only conditional, \dL{}'s congruence rule cannot be used to directly prove the same safety property using $L^{-1}_\text{expl}$.
The safety transfer from one model to the other had to be done manually.
This made the proof exceedingly tedious, with over 200 tactics used.

With local refinements, updating the test from $\ptest{L^{-1}_\text{impl}}$ to $\ptest{L^{-1}_\text{expl}}$ is \emph{one} local refinement which requires a proof of 
\[
(\mathit{init} \land L^{-1}_\text{expl}) \limply \dbox{\prepeat{((\pchoice{\mathit{advisory};\ptest{L^{-1}_\text{expl}}}{\mathit{skip}});\mathit{acc};\mathit{motion})}}{\dbox{\mathit{advisory}}{\mathit{init}}}
\]
The condition $\mathit{init}$ being composed of a simple fact about the advisory and that some constants are nonnegative (thus a trivial loop invariant), proving the safety of the system using the explicit formulation follows easily.
Overall, this decreases the size of the proof to 16 tactics, reducing the need for manual guidance by an order of magnitude.

\subsection{Refactoring Event-triggered Control to Time-triggered Control}
We show how the refactoring techniques presented allow transforming an event-triggered system to a time-triggered system.
The overall result is stronger than \cite{DBLP:conf/lics/LoosP16}, in addition to being streamlined by the refactorings and formally proved using \KeYmaeraX.
The equivalent model $\mathit{event}$ (\rref{fig:model}) from \cite{DBLP:conf/lics/LoosP16} already mentions time, which we circumvent by using the ghost refinement (R2) (\rref{fig:lics}).

Event-triggered systems use events as detection mechanisms, which intrinsically assume a continuous sensing from the system: the system regains control as soon as an event occurs.
This behavior makes them simple to prove safe.
Conversely, time-triggered systems react on a clock, similar to traditional software.
They are thus easier to implement in real systems, but also harder to prove.

Refactoring via refinement enables the transfer of the proof from event-triggered systems to time-triggered systems.
Thus, it makes possible to start designing systems as event-triggered, and iterates until a proof of safety is found, before then tackling the more challenging time-triggered system and reusing the insights developed in the first phase without having to restart the proof from scratch.
Essentially, bridging the gap between the two type of triggers makes implementable systems also easier to prove.

For demonstrating the capabilities of refinements, we will use the example of a car with position $x$ which must preserve some safety distance in front $x \leq \astrm$, and may either brake or accelerate.
The continuous dynamics are written $\mathit{plant}$ and could be simply expressing the link between acceleration and position, i.e., $\D{x}=v,\D{v}=a$.
An event-triggered model for the car may sense the position of $x$ exactly and force the controller to take a decision at the limit, i.e., when $x = \astrm$.
The hybrid program $\mathit{event}$ from \rref{fig:model} corresponds to this event-triggered model.
Having a choice between two continuous dynamics, one with $x \leq \astrm$ and one with $x \geq \astrm$ ensures that going from $x < \astrm$ to $x > \astrm$ requires at least two loop iterations, the first one stopping at $x = \astrm$, while the second depend on the controller decision.
\begin{figure}
\vspace*{-1em}
\begin{align*}
\mathit{event} \defeq & \prepeat{((\pchoice{\mathit{brake}}{\ptest{x < \astrm};\mathit{acc}});(\pchoice{\pode{\mathit{plant}}{x \leq \astrm}}{\pode{\mathit{plant}}{x \geq \astrm}}))}\\
\mathit{time} \defeq & \prepeat{((\pchoice{\mathit{brake}}{\ptest{x < \astrm_T};\mathit{acc}});\pupdate{\umod{t}{0}};\pode{\mathit{plant}, \D{t}=1}{t \leq T})}
\end{align*}
\caption{Event and time-triggered systems}
\label{fig:model}
\end{figure}
\vspace*{-0.5em}
On the other hand, a time-triggered system needs an explicit time $t$ and may only stop at the latest every $T$ seconds.
Thus, for this model to be still correct, it should react preemptively, so the condition for accelerating should be stronger, and depends on $T$.
This is expressed in the hybrid program $\mathit{time}$ from \rref{fig:model}.

As a special case, we will consider the situation where both accelerating and braking are constant -- with value $A$, $-B$ respectively -- and the car needs to stop before a position $m$, meaning the postcondition is $x \leq m$.
In that case, $\astrm \defeq m - \frac{v^2}{2B}$ and $\astrm_T \defeq m - \frac{v^2}{2B} - (\frac{A}{B} + 1)(\frac{A}{2}\cdot t^2 + t v)$ lead to safe systems assuming all constants are positive.

A proof $\asfml \limply \dbox{\mathit{event}}{\bsfml}$ for the event model can be transformed into a proof to the time model $\asfml \limply \dbox{\mathit{time}}{\bsfml}$ via the refactorings in \rref{fig:lics}.
\begin{figure}
  \centering
\begin{tikzpicture}
  \node (a) at (0,0) {\(\prepeat{((\pchoice{\mathit{brake}}{\ptest{x < \astrm};\mathit{acc}});(\pchoice{\pode{\mathit{plant}}{x \leq \astrm}}{\pode{\mathit{plant}}{x \geq \astrm}}))}\)};
  
  \node (b) at (0,-1.3) {\(\prepeat{((\pchoice{\mathit{brake}}{\ptest{x < \astrm};\mathit{acc}});\pode{\mathit{plant}}{x \leq \astrm})}\)};
  \draw[->] (a) -- (b) node[midway, right] {global refinement (R1)};
  
  \node (c) at (0,-2.6) {\(\prepeat{((\pchoice{\mathit{brake}}{\ptest{x < \astrm};\mathit{acc}});\pupdate{\pumod{t}{0}};\pode{\mathit{plant},\D{t}=1}{x \leq \astrm})}\)};
  \draw[->] (b) -- (c) node[midway, right] {ghost refinement (R2)};

  \node (d) at (0,-3.9) {\(\prepeat{((\pchoice{\mathit{brake}}{\ptest{x < \astrm_T};\mathit{acc}});\pupdate{\pumod{t}{0}};\pode{\mathit{plant},\D{t}=1}{x \leq \astrm})}\)};
  \draw[->] (c) -- (d) node[midway, right] {local refinement (R3)};

  \node (e) at (0,-5.2) {\(\prepeat{((\pchoice{\mathit{brake}}{\ptest{x < \astrm_T};\mathit{acc}});\pupdate{\pumod{t}{0}};\pode{\mathit{plant},\D{t}=1}{t \leq T})}\)};
  \draw[->] (d) -- (e) node[midway, right] {local refinement (R4)};
\end{tikzpicture}
\caption{Event-triggered to time-triggered system}
\label{fig:lics}
\end{figure}
The first refactoring (R1) removes one of the continuous dynamics.
This is done by a single global refinement using \irref{choiceL}.
Refactoring (R2) introduces the new time variable $t$ in the differential equation.
It corresponds to a ghost refinement using \irref{refDG}.
Compared to (R1), this step will only succeed if $t$ and $t'$ are correctly fresh with regards to the system (e.g., $\mathit{plant}$ does not use $t$) and the postcondition $\bsfml$.
Refactoring (R3) updates the test condition to take into account for the potential delay $T$.
It requires as input a proof that $\astrm_T \leq \astrm$ which is unlikely to hold in general.
Thus, we fall back to local refinement, and need to prove that this inequality holds locally.
For our specific application, this leads to proving \(A > 0 \land B > 0 \land T > 0 \limply \astrm_T < \astrm\) which is simple enough to be proved automatically using tools for first-order arithmetic.
This lemma proved, local refinement justifies the refactoring if \(A > 0 \land B > 0 \land T > 0\) holds locally, i.e., it is an invariant of the system.
Since these are all constants, the proof is trivial.
All the manual reasoning above for (R3) only needed to be done on a local basis, as the extension to the full system is automatic.
Refactoring (R4) is similar but updates the domain constraint.
Since the change occurs inside a differential equation, the resulting lemma is more complex.
Indeed, we must show that \(\dbox{\pode{\mathit{plant}, \D{t}=1}{t \leq T}} x \leq \astrm\) holds locally, i.e., after running the updated controller.
This is the essence of the transformation from event-triggered to time-triggered, which has been distilled from the general proof via the automatic congruence reasoning.

Due to the complex handling of nondeterministic assignments for ghost refinements, (R2) has more than half the total number of internal steps used, but the automation ensures it only requires 5-10 manual steps similar to the others.

\section{Conclusion}
We showed how refinements lead to different notions of formally proved refactoring for hybrid programs.
The flexibility of \dRL's refactoring-as-propositions principle enables proofs to go beyond global reasoning and to have access to local assumptions when proving refinement of subprograms, without drastically extending its axiomatization.
By adding final nondeterministic assignments, we have shown how refinements can suppress the value of some variables, opening the path to ghost refinement as a formal base to allow introduction of fresh variables soundly.
These refactoring techniques via refinements are implemented in the theorem prover \KeYmaeraX.
Since the refactorings presented here rely on refinements for their soundness, they cannot be used directly to prove refactorings that introduce more behaviors to a safe system.
Integrating these refactorings formally would require to use the current techniques to isolate the new behaviors---without affecting its overall execution---so that safety is proved by simply looking at the isolated change rather than the full system.
Future works include supporting refinements of more complex programs of \dL extensions such as communications \cite{DBLP:conf/cade/BriegerMP23}, which are also used in cyber-physical systems.

\begin{credits}
\subsubsection{\ackname}
This work has been supported by an Alexander von Humboldt Professorship, the pilot program Core Informatics (KiKIT) of the Helmholtz Association (HGF), and the Deutsche Forschungsgemeinschaft (DFG, German Research Foundation) - SFB 1608 - 501798263.

\subsubsection{\discintname}
The authors have no competing interests to declare that are
relevant to the content of this article.
\end{credits}

\bibliographystyle{splncs04}
\bibliography{platzer,dRL}

\appendix
\section{Additional Rules and Axioms used}
\label{app:axs}
\begin{calculuscollection}
\begin{calculus}
  \cinferenceRule[randtestmerge|${:}*_{\text{merge}}$]{Merging randoms with test}
  {\linferenceRule[eq]
    {\prandom{x};\ptest{\lexists y \ousfml[y]}}
    {\axkey{\prandom{x};\ptest{\ousfml[x]};\prandom{x}}}
  }{}    

  \cinferenceRule[leqantisym|$\prgeq$]{Antisymmetry of $\refines$}
  {\linferenceRule[equiv]
    {\ausprg \refines \busprg \land \busprg \refines \ausprg}
    {\axkey{\ausprg \prgeq \busprg}}
  }{}
\end{calculus}
\qquad
\begin{calculus}
  \cinferenceRule[MPax|MP]{modus ponens}
  {\linferenceRule[formula]
    {\asfml\limply \bsfml \quad \asfml}
    {\bsfml}
  }{}

  \cinferenceRule[equivifyr|${\limply}2{\lbisubjunct}$]{equivify right}
  {\linferenceRule[sequent]
    {\lsequent{}{\asfml\lbisubjunct\bsfml}}
    {\lsequent{}{\asfml\limply\bsfml}}
  }{}    
\end{calculus}
\\
\begin{calculus}
\cinferenceRuleQuoteDef{id}

\cinferenceRuleQuoteDef{implyr}

\cinferenceRuleQuoteDef{andr}

\cinferenceRuleQuoteDef{andl}
\cinferenceRuleQuoteDef{dW}
\end{calculus}
\qquad
\begin{calculus}
\cinferenceRuleQuoteDef{orr}

\cinferenceRuleQuoteDef{weakenl}
\cinferenceRuleQuoteDef{weakenr}
\cinferenceRule[cut|cut]{cut}
{\linferenceRule[sequent]
  {\lsequent[L]{}{\csfml}
  & \lsequent[L]{\csfml}{}}
  {\lsequent[L]{}{}}
}{}%
\end{calculus}
\\
\begin{calculus}
\cinferenceRuleQuoteDef{band}
\cinferenceRuleQuoteDef{randomd}
\cinferenceRuleQuoteDef{K}
\cinferenceRuleQuoteDef{DE}
\cinferenceRuleQuoteDef{DG}
\end{calculus}
\\
\begin{calculus}
\cinferenceRuleQuoteDef{choicel}
\cinferenceRuleQuoteDef{update}
\cinferenceRuleQuoteDef{seqidr}
\cinferenceRuleQuoteDef{loopl}
\cinferenceRuleQuoteDef{unloop}
\end{calculus}
\quad
\begin{calculus}
\cinferenceRuleQuoteDef{choicer}
\cinferenceRuleQuoteDef{assignmerge}
\cinferenceRule[skip|${:}{*}_{\text{skip}}$]{Skip random}
{\linferenceRule[leq]
  {\prandom{x}}
  {\ptest{\ltrue}}
}{}%
\cinferenceRuleQuoteDef{loopr}
\cinferenceRuleQuoteDef{unfold-l}
\end{calculus}
\\
\begin{calculus}
\cinferenceRuleQuoteDef{assigndet}
\cinferenceRuleQuoteDef{ode}
\end{calculus}
\end{calculuscollection}

\section{Proofs for \rref{sec:glob_ref} and \rref{sec:loc_ref}}\label{app:312}

\irlabel{prop|prop}
\irlabel{R|Real}
\begin{proof}[\rref{thm:useAt}]
  The proof use axioms in the form of implication or equivalence as a rule.
  \rref{fig:ax_rule} shows how the implicit rule on the left expands to the concrete proof via \irref{MPax} (and \irref{equivifyr}).

  \begin{figure}
  \begin{minipage}{0.20\textwidth}
  \begin{sequentdeduction}[array]
    \linfer[Ax]{\lsequent[L]{}{\bsfml}}{\lsequent[L]{}{\asfml}}
  \end{sequentdeduction}
  \end{minipage}
  \begin{minipage}{0.45\textwidth}
  \begin{sequentdeduction}[array]
    \linfer[MPax]{
      \linfer[Ax]{}{\lsequent{}{\bsfml\limply\asfml}}
      !\lsequent[L]{}{\bsfml}
    }{\lsequent[L]{}{\asfml}}
  \end{sequentdeduction}
  \end{minipage}
  \begin{minipage}{0.45\textwidth}
  \begin{sequentdeduction}[array]
    \linfer[MPax]{
      \linfer[equivifyr]{
        \linfer[Ax]{}{\lsequent{}{\bsfml\lbisubjunct\asfml}}
      }{\lsequent{}{\bsfml\limply\asfml}}
      !\lsequent[L]{}{\bsfml}
    }{\lsequent[L]{}{\asfml}}
  \end{sequentdeduction}
  \end{minipage}
  \caption{Axiom as rule with the corresponding proof}\label{fig:ax_rule}
  \end{figure}

  The derivation is done by induction on the context $C_{k_1,k_2}^{\pm}$.
  The inference does not depend on the type of hole ($k_2$), only the constructors and the polarities, thus the inference used generic names -- $\asprg_i,\bsprg_i$ for programs, $\asfml_i,\bsfml_i$ for formulas -- to increase readability.
  The case for first-order logic constructs are standard and are omitted.
  In line with the implementation, rather than implications $\asfml \limply \bsfml$, we use sequents $\lsequent{\asfml}{\bsfml}$ instead, which are equivalent (by \irref{MPax} and \irref{implyr}).

  \begin{itemize}
  \item $\cdot_p$: the inference is straightforward as $\contextapp{\cdot_p}{\asprg} = \asprg$.
    \[\linfer{\asprg \refines \bsprg}{\contextapp{\cdot_p}{\asprg} \refines \contextapp{\cdot_p}{\bsprg}}\]
  
  \item $\asprg$ (no hole):
  \[\linfer[leqrefl]{}{\asprg \refines \asprg}\]
  
  \item $\ptest{C_{fk}^\pm}$:
  \begin{sequentdeduction}[array]
    \linfer[test]{
      \linfer[implyr]{
        \lsequent{\asfml}{\bsfml}
      }\lsequent{}{\asfml \limply \bsfml}
    }{\lsequent{}{\ptest{\asfml} \refines \ptest{\bsfml}}}
  \end{sequentdeduction}
  
  \item $C_{pk}^\pm;C_{pk}^\pm$:
  \begin{sequentdeduction}[array]
    \linfer[sequence]{
      \linfer[andr]{
        \lsequent{}{\asprg_1 \refines \bsprg_1}
        ! \linfer[G]{
          \lsequent{}{\asprg_2 \refines \bsprg_2}
        }{\lsequent{}{\dbox{\asprg_1}(\asprg_2 \refines \bsprg_2)}}
      }{\lsequent{}{\asprg_1 \refines \bsprg_1 \land \dbox{\asprg_1}\asprg_2\refines \bsprg_2}}
    }\lsequent{}{\asprg_1;\asprg_2 \refines \bsprg_1;\bsprg_2}
  \end{sequentdeduction}
  
  \item $\pchoice{C_{pk}^\pm}{C_{pk}^\pm}$:
  \begin{sequentdeduction}[array]
    \linfer[choicel]{
      \linfer[andr]{
        \linfer[choicer]{
          \linfer[orr]{
            \linfer[weakenr]{
              \lsequent{}{\asprg_1 \refines \bsprg_1}
            }{\lsequent{}{\asprg_1 \refines \bsprg_1, \asprg_1 \refines \bsprg_2}}
          }{\lsequent{}{\asprg_1 \refines \bsprg_1 \lor \asprg_1 \refines \bsprg_2}}
        }{\lsequent{}{\asprg_1 \refines \pchoice{\bsprg_1}{\bsprg_2}}}
        !\linfer[choicer]{
          \linfer[orr]{
            \linfer[weakenr]{
              \lsequent{}{\asprg_2 \refines \bsprg_2}
            }{\lsequent{}{\asprg_2 \refines \bsprg_1, \asprg_2 \refines \bsprg_2}}
          }{\lsequent{}{\asprg_2 \refines \bsprg_1 \lor \asprg_2 \refines \bsprg_2}}
        }{\lsequent{}{\asprg_2 \refines \pchoice{\bsprg_1}{\bsprg_2}}}
      }{\lsequent{}{\asprg_1 \refines \pchoice{\bsprg_1}{\bsprg_2} \land \asprg_2\refines \pchoice{\bsprg_1}{\bsprg_2}}}
    }{\lsequent{}{\pchoice{\asprg_1}{\asprg_2}\refines \pchoice{\bsprg_1}{\bsprg_2}}}
  \end{sequentdeduction}
  
  \item $\pode{\D{x}=\astrm}{C_{fk}^\pm}$:
  \begin{sequentdeduction}[array]
    \linfer[refOde]{
      \linfer[band]{
        \linfer[andr]{
          \linfer[DE]{
            \linfer[G]{
              \linfer[assignb]{
                \linfer[R]{}{\lsequent{}{\astrm = \astrm}}
              }{\lsequent{}{\dbox{\umod{\D{x}}{\astrm}}\D{x}=\astrm}}
            }{\lsequent{}{\dbox{\pode{\D{x}=\astrm}{\asfml}}{\dbox{\umod{\D{x}}{\astrm}}\D{x}=\astrm}}}
          }{\lsequent{}{\dbox{\pode{\D{x}=\astrm}{\asfml}}{\D{x}=\astrm}}}
          !
          \linfer[dW]{
            \lsequent{\asfml}{\bsfml}
          }{\lsequent{}{\dbox{\pode{\D{x}=\astrm}{\asfml}}{\bsfml}}}
        }{\lsequent{}{\dbox{\pode{\D{x}=\astrm}{\asfml}}{\D{x}=\astrm} \land \dbox{\pode{\D{x}=\astrm}{\asfml}}{\bsfml}}}
      }{\lsequent{}{\dbox{\pode{\D{x}=\astrm}{\asfml}}{(\D{x}=\astrm \land \bsfml)}}}
    }{\lsequent{}{\pode{\D{x}=\astrm}{\asfml} \refines \pode{\D{x}=\astrm}{\bsfml}}}
  \end{sequentdeduction}
  
  \item $\prepeat{C_{pk}^\pm}$:
  \begin{sequentdeduction}
    \linfer[unloop]{
      \linfer[G]{
        \lsequent{}{\asprg \refines \bsprg}
      }{\lsequent{}{\dbox{\prepeat{\asprg}}{\asprg \refines \bsprg}}}
    }{\lsequent{}{\prepeat{\asprg} \refines \prepeat{\bsprg}}}
  \end{sequentdeduction}
  
  \item $\dbox{C_{pk}^\mp}{C_{fk}^\pm}$:
  \begin{sequentdeduction}[array]
    \linfer[MPax]{
      \linfer[boxref]{
        \linfer[weakenl]{
          \lsequent{}{\bsprg\refines\asprg}
        }{\lsequent{\dbox{\asprg}{\asfml}}{\bsprg\refines\asprg}}
      }{\lsequent{\dbox{\asprg}{\asfml}}{\dbox{\asprg}{\bsfml}\limply\dbox{\bsprg}{\bsfml}}}
      !
      \linfer[MPax]{
        \linfer[K]{
          \linfer[weakenl]{
            \linfer[G]{
              \linfer[implyr]{
                \lsequent{\asfml}{\bsfml}
              }{\lsequent{}{\asfml \limply \bsfml}}
            }{\lsequent{}{\dbox{\asprg}{\asfml \limply \bsfml}}}
          }{\lsequent{\dbox{\asprg}{\asfml}}{\dbox{\asprg}{\asfml \limply \bsfml}}}
        }{\lsequent{\dbox{\asprg}{\asfml}}{\dbox{\asprg}{\asfml}\limply\dbox{\asprg}{\bsfml}}}
        !
        \linfer[id]{}{
          \lsequent{\dbox{\asprg}{\asfml}}{\dbox{\asprg}{\asfml}}
        }
      }{\lsequent{\dbox{\asprg}{\asfml}}{\dbox{\asprg}{\bsfml}}}
    }{\lsequent{\dbox{\asprg}{\asfml}}{\dbox{\bsprg}{\bsfml}}}
  \end{sequentdeduction}
  
  \item $C_{pk}^\mp \refines C_{pk}^\pm$:
  
  \hspace*{-1cm}
  \begin{minipage}{1\textwidth}
  \renewcommand{\linferPremissSeparation}{\hspace{0.1cm}}
  \begin{sequentdeduction}[array]
    \linfer[MPax]{
      \linfer[leqtrans]{
        \linfer[MPax]{
          \linfer[leqtrans]{
            \linfer[weakenl]{
              \lsequent{}{\bsprg_1\refines\asprg_1}
            }{\lsequent{\asprg_1\refines\asprg_2}{\bsprg_1\refines\asprg_1}}
          }{\lsequent{\asprg_1\refines\asprg_2}{\asprg_1\refines\asprg_2 \limply \bsprg_1\refines \asprg_2}}
          !
          \linfer[id]{}{
              \lsequent{\asprg_1\refines\asprg_2}{\asprg_1\refines\asprg_2}
          }
        }\lsequent{\asprg_1\refines\asprg_2}{\bsprg_1 \refines \asprg_2}
      }{\lsequent{\asprg_1\refines\asprg_2}{\asprg_2\refines\bsprg_2 \limply \bsprg_1\refines \bsprg_2}}
      ! \linfer[weakenl]{
        \lsequent{}{\asprg_2\refines\bsprg_2}
        }\lsequent{\asprg_1\refines\asprg_2}{\asprg_2\refines\bsprg_2}
    }{\lsequent{\asprg_1\refines \asprg_2}{\bsprg_1\refines \bsprg_2}}
  \end{sequentdeduction}
  \end{minipage}
  
  \end{itemize}
\end{proof}

Having proved \rref{thm:useAt}, the same rules as \rref{fig:ax_rule} may now be used within contexts, that is:
\begin{sequentdeduction}[array]
    \linfer[Ax]{\lsequent[L]{}{C_{ff}^+(\bsfml)}}{\lsequent[L]{}{C_{ff}^+(\asfml)}}
\end{sequentdeduction}
Proofs using refinements often benefit from introducing intermediate programs, where each one is related to the previous one by the use of an axiom (with sometimes \rref{thm:useAt} implicitly).
For that purpose, we will use the following notation:
\begin{align*}
	\asprg &\prgeq \bsprg \text{ by }(C)\\
	&\refines \csprg \text{ by }(D)
\end{align*}
This means that $\asprg \prgeq \bsprg$ is derivable by the axiom $(C)$ and $\bsprg \refines \csprg$ is derivable by the axiom $(D)$.
The derivation of $\asprg \refines \csprg$ is then obtained via the transitivity of the refinement relation (\irref{leqtrans}).
As both sequential composition and choice are associative, such rewriting will be omitted for clarity.
For instance, we can derive that nondeterministic assignment is idempotent (\irref{randidem}):
\begin{align*}
  \prandom{x};\prandom{x} &\prgeq \prandom{x};\ptest{\ltrue};\prandom{x} & \text{by }\irref{seqidr}\\
  &\prgeq \prandom{x};\ptest{\lexists x \ltrue} & \text{by }\irref{randtestmerge}\\
  &\prgeq \prandom{x};\ptest{\ltrue} & \text{by FOL}\\
  &\prgeq \prandom{x} & \text{by }\irref{seqidr}
\end{align*}
\begin{calculus}
    \cinferenceRule[randidem|${:=}{*}_{\text{idem}}$]{Random idempotent}
  {\linferenceRule[eq]
    {\prandom{x}}
    {\prandom{x};\prandom{x}}
  }{}    
\end{calculus}

This notation will be used extensively for the proof of \rref{thm:fresh}.

\begin{proof}[\rref{lem:cprge}]
  Given our syntax, if a context $C_{fp}$ has a single hole $\cdot_p$, then it always has a polarity.
  For instance, if $C_{fp}$ is monotone, then by \rref{thm:useAt}, \(C_{fp}(\asprg) \limply C_{fp}(\bsprg)\) can be derived from \(\asprg \refines \bsprg\) but also \(C_{fp}(\bsprg) \limply C_{fp}(\asprg)\) can be derived from \(\bsprg \refines \asprg\).
  Thus, the rule \irref{CPrgE} is derived using antisymmetry of program equivalence and formula equivalence.

  In case the context $C_{fp}$ has $n > 1$ holes, each hole is numbered using $0 \leq i < n$.
  For $0 \leq i < n$, we define $C_{fp}^i$ as the single hole context where each hole $j < i$ is replaced by $\bsprg$, and each hole $j > i$ is replaced by $\asprg$.
  In particular, we have $C_{fp}^0(\asprg) = C_{fp}(\asprg)$, for $0 < i < n$, $C_{fp}^{i}(\asprg) = C_{fp}^{i-1}(\bsprg)$, and $C_{fp}^{n-1}(\bsprg) = C_{fp}(\bsprg)$.
  Thus, by transitivity, it boils down to the single hole case used repeatedly on $C_{fp}^i(\asprg), C_{fp}^i(\bsprg)$ for all $i < n$.
\end{proof}

\begin{proof}[\rref{lem:focus}]
We prove the equivalent result that the sequent \(\lsequent{C_\asprg(\asprg \refines \bsprg)}{\contextapp{C}{\asprg} \refines \contextapp{C}{\bsprg}}\) is derivable.
It derives the rule from \rref{lem:focus} as follows:
\begin{sequentdeduction}[array]
  \linfer[cut]{
    \linfer[weakenl]{
      \linfer{}{
        \lsequent{C_\asprg(\asprg \refines \bsprg)}{\contextapp{C}{\asprg} \refines \contextapp{C}{\bsprg}}
      }
    }{\lsequent{\Gamma,C_\asprg(\asprg \refines \bsprg)}{\contextapp{C}{\asprg} \refines \contextapp{C}{\bsprg}}}
    !
    \linfer[weakenr]{
      \lsequent{\Gamma}{C_\asprg(\asprg \refines \bsprg)}
    }{\lsequent{\Gamma}{\contextapp{C}{\asprg} \refines \contextapp{C}{\bsprg}, C_\asprg(\asprg \refines \bsprg)}}
  }{\lsequent{\Gamma}{\contextapp{C}{\asprg} \refines \contextapp{C}{\bsprg}}}
\end{sequentdeduction}
For constructors that do not affect the projective context $C_\asprg$, the inferences are similar to the ones of \rref{thm:useAt}, while preserving the context $C_\asprg(\asprg \refines \bsprg)$.
For the sequence and the loop however, the rule \irref{G} removes all existing context.
As the projective context contains the same box modality as the one removed by \irref{G}, the inference can be adapted as follows -- note that for the loop case, we have $C(\asprg)$ instead of $\csprg$ though the inference is identical.
\begin{sequentdeduction}[array]
  \hspace*{-1cm}
  \linfer[MPax]{
    \linfer[weakenl]{
      \linfer[K]{
        \linfer[G]{
          \linfer[implyr]{
            \lsequent{C_\asprg(\asprg \refines \bsprg)}{\contextapp{C}{\asprg} \refines \contextapp{C}{\bsprg}}
          }{\lsequent{}{C_\asprg(\asprg \refines \bsprg) \limply \contextapp{C}{\asprg} \refines \contextapp{C}{\bsprg}}}
        }{\lsequent{}{\dbox{\csprg}{(C_\asprg(\asprg \refines \bsprg) \limply \contextapp{C}{\asprg} \refines \contextapp{C}{\bsprg})}}}
      }{\lsequent{}{\dbox{\csprg}{C_\asprg(\asprg \refines \bsprg)} \limply \dbox{\csprg}{\contextapp{C}{\asprg} \refines \contextapp{C}{\bsprg}}}}
    }{\lsequent{\dbox{\csprg}{C_\asprg(\asprg \refines \bsprg)}}{\dbox{\csprg}{C_\asprg(\asprg \refines \bsprg)} \limply \dbox{\csprg}{\contextapp{C}{\asprg} \refines \contextapp{C}{\bsprg}}}}
    !
    \linfer[id]{}{
      \lsequent{\dbox{\csprg}{C_\asprg(\asprg \refines \bsprg)}}{\dbox{\csprg}{C_\asprg(\asprg \refines \bsprg)}}
    }
  }{\lsequent{\dbox{\csprg}{C_\asprg(\asprg \refines \bsprg)}}{\dbox{\csprg}{\contextapp{C}{\asprg} \refines \contextapp{C}{\bsprg}}}}
\end{sequentdeduction}
\end{proof}

\section{Proofs for \rref{sec:ghost}}\label{app:33}

Regarding the proofs that new axioms are sound, we refer to the standard semantics of \dRL{}\cite{DBLP:conf/ijcar/PrebetP24}.
\begin{proof}[\rref{lem:ghost_ode}]
  \begin{itemize}
    \item \irref{odeconst}:
    The axiom can be derived directly from the logic.
    From the side conditions, we have that $x$ is not free in the equivalence.
    
    \hspace*{-2cm}
    \begin{minipage}{\textwidth}
    \begin{sequentdeduction}[array]
      \linfer[assignb]{
        \linfer[assigndet]{
          \linfer[sequence]{
            \linfer[andr]{
              \linfer[assignmerge]{\lclose}{
                \lsequent{}{\umod{x}{\astrm} \prgeq \umod{x}{\astrm};\umod{x}{\astrm}}
              }
            ! \linfer[G]{
                \linfer[leqrefl]{\lclose}{
                  \lsequent{}{\pode{\D{y}=\astrm}{\bsfml};\prandom{x} \prgeq \pode{\D{y}=\astrm}{\bsfml};\prandom{x}}
                }
              }{\lsequent{}{\dbox{\umod{x}{\astrm}}{(\pode{\D{y}=\astrm}{\bsfml};\prandom{x} \prgeq \pode{\D{y}=\astrm}{\bsfml};\prandom{x})}}}
            }{\lsequent{}{\umod{x}{\astrm} \prgeq \umod{x}{\astrm};\umod{x}{\astrm} \land \dbox{\umod{x}{\astrm}}{(\pode{\D{y}=\astrm}{\bsfml};\prandom{x} \prgeq \pode{\D{y}=\astrm}{\bsfml};\prandom{x})}}}
          }{\lsequent{}{\umod{x}{\astrm};\pode{\D{y}=\astrm}{\bsfml};\prandom{x} \prgeq \umod{x}{\astrm};\umod{x}{\astrm};\pode{\D{y}=x}{\bsfml};\prandom{x}}}
        }{\lsequent{}{\dbox{\umod{x}{\astrm}}{(\pode{\D{y}=\astrm}{\bsfml};\prandom{x} \prgeq \umod{x}{\astrm};\pode{\D{y}=x}{\bsfml};\prandom{x})}}}
      }{\lsequent{}{\pode{\D{y}=\astrm}{\bsfml};\prandom{x} \prgeq \umod{x}{\astrm};\pode{\D{y}=x}{\bsfml};\prandom{x}}}
    \end{sequentdeduction}
    \end{minipage}
    
    \item \irref{refDG}:
    We prove the axiom by proving the equivalent formula \cite{DBLP:conf/lics/LoosP16}:
    \begin{align*}
      \lforall{x^+}\lforall{y^+}(&\ddiamond{\pupdate{\umod{x}{c}};\pode{\D{y}=\astrm,\D{x}=\bstrm}{\bsfml};\prandom{x};\prandom{\D{x}}}{(x = x^+ \land y = y^+)}\\
      &\lbisubjunct \ddiamond{\pode{\D{y}=\astrm}{\bsfml};\prandom{x};\prandom{\D{x}}}{(x = x^+ \land y = y^+)})
    \end{align*}
    The formula $\ddiamond{\prandom{x};\prandom{\D{x}}}{(x = x^+ \land y = y^+)}$ reduces by \irref{randomd} to the first-order formula: $\lexists{x}\lexists{x'}{(x = x^+ \land y = y^+)}$ which is equivalent to $y = y^+$.
    So by duality, our axiom \irref{refDG} is equivalent to:
    \[\lforall{y^+}(\dbox{\pupdate{\umod{x}{c}};\pode{\D{y}=\astrm,\D{x}=\bstrm}{\bsfml}}y \neq y^+ \lbisubjunct \dbox{\pode{\D{y}=\astrm}{\bsfml}}y \neq y^+)\]
    We prove both directions using the fact that axiom \irref{DG} is also sound when replacing the existential quantifier by a universal quantifier \cite{DBLP:journals/jar/Platzer17}.
    \begin{itemize}
      \item $\dbox{\pupdate{\umod{x}{c}};\pode{\D{y}=\astrm,\D{x}=\bstrm}{\bsfml}}y \neq y^+$ implies \(\lexists{x}\dbox{\pode{\D{y}=\astrm,\D{x}=\bstrm}{\bsfml}}y \neq y^+\) by choosing $c$ as value for $x$.
      By \irref{DG}, we thus have $\dbox{\pode{\D{y}=\astrm}{\bsfml}}y \neq y^+$.
      \item Assuming $\dbox{\pode{\D{y}=\astrm}{\bsfml}}y \neq y^+$, by using the variant of \irref{DG} with universal quantification, we have \(\lforall{x}\dbox{\pode{\D{y}=\astrm,\D{x}=\bstrm}{\bsfml}}y \neq y^+\).
      In particular, the property holds for $x = c$ and thus $\dbox{\pupdate{\umod{x}{c}};\pode{\D{y}=\astrm,\D{x}=\bstrm}{\bsfml}}y \neq y^+$.
    \end{itemize}
  \end{itemize}
\end{proof}

The derivation of \rref{thm:fresh} requires the addition of two new axioms.
\irref{DErefL} highlights that ODEs overwrite the value of differential variables, while \irref{refRandomComm} enables a quicker implementation, by swapping nondeterministic assignment easily.

\vspace{0.2em}
\begin{calculus}
  \cinferenceRule[DErefL|$\text{DE}_{\prgeq l}$]{DE left axiom}
{\linferenceRule[eq]
  {\prandom{\D{x}};\pode{\D{x}=\oustrm[x]}{\ousfml[x]}}
  {\pode{\D{x}=\oustrm[x]}{\ousfml[x]}}
}{}
  \cinferenceRule[refRandomComm|{${{:=}{*}}_{\text{swap}}$}]{Nondet commute}
{\linferenceRule[eq]
  {\prandom{x};\asprg}
  {\asprg;\prandom{x}}
}{$x$ fresh for $\asprg$}
\end{calculus}

\begin{lemma}\label{lem:rand}
The formulas \irref{DErefL} and \irref{refRandomComm} is valid.
\end{lemma}
\begin{proof}~
\begin{itemize}
  \item \irref{DErefL}:
      The axiom follow from the fact that a solution $\sol$ to the differential equation must satisfies
      \begin{itemize}
        \item $\sol(0)$ is a $\set{\D{x}}$-variation of the initial state, i.e., it is equal to the initial state for all variables but $\D{x}$.
        So the solution is not affected by the nondeterministic assignment on $\D{x}$.
        \item $\ivaluation{\Isol}{x'} = \ivaluation{\Isol}{f(x)}$ for all $t$, so the value assigned to $\D{x}$ by the nondeterministic assignment is overwritten by the solution $\sol$.
      \end{itemize}

  \item \irref{refRandomComm}:
  As $x$ is fresh for $\asprg$, in particular, $x \notin \freevars{\asprg}$ and $x \notin \boundvars{\asprg}$.
  This means we can use the coincidence property for free variables and the bound effect for programs \cite{DBLP:conf/ijcar/PrebetP24}.
  We note $\modif{\iget[state]{\Itt}}{x}{r}$ with $r\in\R$ for the state where $\ivaluation{\It}y = \ivaluation{\Itt}y$ for all variables $y \neq x$ and $\ivaluation{\It}x = r$.
  The coincidence property implies that if $\iaccessible[\asprg]{\I}{\It}$, then for all $r\in\R$, there exists $r'\in\R$ such that $\iaccessible[\asprg]{\imodif[state]{\I}{x}{r}}{\imodif[state]{\It}{x}{r'}}$.
  On the other hand, the bound effect implies that if $\iaccessible[\asprg]{\I}{\It}$, then $\ivaluation{\I}{x} = \ivaluation{\It}{x}$.
  Thus, we can further refine the coincidence property, ensuring that $r' = r$, i.e., $\iaccessible[\asprg]{\imodif[state]{\I}{x}{r}}{\imodif[state]{\It}{x}{r}}$ if and only if $\iaccessible[\asprg]{\I}{\It}$.

  Take $\iaccessible[\asprg;\prandom{x}]{\I}{\It}$, then there exists $\iget[state]{\Itt}$ such that $\iaccessible[\asprg]{\I}{\Itt}$ and $\iaccessible[\prandom{x}]{\Itt}{\It}$.
  As the latter program only affects $x$ nondeterministically, we have $\iget[state]{\It} = \modif{\iget[state]{\Itt}}{x}{r}$ for some $r\in\R$.
  But then we have $\iaccessible[\prandom{x}]{\I}{\imodif[state]{\I}{x}{r}}$ and $\iaccessible[\asprg]{\imodif[state]{\I}{x}{r}}{\imodif[state]{\Itt}{x}{r}}$, so the program $\prandom{x};\asprg$ can also reach the state $\iget[state]{\It}$ from the state $\iget[state]{\I}$.

  Conversely, take $\iaccessible[\prandom{x};\asprg]{\I}{\It}$, then there exists $r\in\R$ such that we have $\iaccessible[\prandom{x}]{\I}{\imodif[state]{\I}{x}{r}}$ and $\iaccessible[\asprg]{\imodif[state]{\I}{x}{r}}{\It}$.
  But then we have $\iaccessible[\asprg]{\imodif[state]{\I}{x}{\nu(x)}}{\imodif[state]{\It}{x}{\nu(x)}}$ and $\iaccessible[\prandom{x}]{\imodif[state]{\It}{x}{\nu(x)}}{\imodif[state]{\It}{x}{\omega(x)}}$ with $\iget[state]{\I} = \modif{\iget[state]{\I}}{x}{\nu(x)}$ and $\iget[state]{\It} = \modif{\iget[state]{\It}}{x}{\omega(x)}$, so the program $\asprg;\prandom{x}$ can also reach the state $\iget[state]{\It}$ from the state $\iget[state]{\I}$.
\end{itemize}
\end{proof}

Due to the handling of loops when proving \rref{thm:fresh}, we first show the following:
\begin{lemma}\label{lem:notfree}
  If $x$ is not free in $\asprg$, then the formula \((\prandom{x};\asprg;\prandom{x}) \prgeq \asprg;\prandom{x}\) is derivable.
\end{lemma}
\begin{proof}
  We reason by induction on $\asprg$.
  \begin{itemize}
    \item $\asprg \equiv \ptest{\asfml}$:
    By assumption, $x \notin \freevars{\asfml}$, so $x$ is fresh for $\asfml$ (up-to alpha renaming).
    
    Thus, we can prove:
    \begin{align*}
      \prandom{x};\ptest{\asfml};\prandom{x} &\prgeq {\ptest{\asfml}};\prandom{x};\prandom{x} & \text{by }(\irref{refRandomComm})\\
      &\prgeq {\ptest{\asfml}};\prandom{x} & \text{by }(\irref{randidem})
    \end{align*}
    
    \item $\asprg \equiv \umod{y}{\astrm}$:
    By assumption, $x \notin \freevars{\astrm}$.
    There are two cases:
    \begin{itemize}
      \item $x \neq y$: in that case, $x$ is free for $\umod{y}{\astrm}$ so we can conclude as above.
      \item $x = y$: 
      \begin{align*}
        \prandom{x};\umod{x}{\astrm};\prandom{x} &\prgeq \prandom{x};\prandom{x};\ptest{x = \astrm};\prandom{x} & \text{by }\irref{update}\\
        &\prgeq \prandom{x};\ptest{x = \astrm};\prandom{x} & \text{by }\irref{randidem}\\
        &\prgeq \umod{x}{\astrm};\prandom{x} & \text{by }\irref{update}
      \end{align*}
    \end{itemize}

    In particular, the proof extends similarly to nondeterministic assignments.
    For the second case, it also extends to any program of the form $\umod{x}{\astrm};\bsprg$.
    
    \item $\asprg \equiv \pode{\D{y}=\astrm}{\asfml}$:
    By assumption $x \neq y, x \notin\freevars{\astrm}$, and $x \notin\freevars{\asfml}$, so either $x$ is fresh for $\asprg$ and the same reasoning as for $\ptest{\asfml}$ holds, or $x$ is the differential variable $\D{y}$.
    In the latter case, we can use \irref{DErefL} to rewrite $\pode{\D{y}=\astrm}{\asfml}$ as $\umod{\D{y}=\astrm};\pode{\D{y}=\astrm}{\asfml}$, and thus going back to the assignment case $\umod{x}{\astrm};\bsprg$.
    come back to the assignment case.

    \item $\asprg \equiv \pchoice{\bsprg}{\csprg}$:
    By assumption, $x \notin\freevars{\bsprg}$ and $x \notin\freevars{\csprg}$ both hold, so we can use both induction hypotheses.
    \begin{align*}
      \prandom{x};(\pchoice{\bsprg}{\csprg});\prandom{x} &\prgeq \prandom{x};(\pchoice{(\bsprg;\prandom{x})}{(\csprg;\prandom{x})}) &\text{by }\irref{seqdistr}\\
      &\prgeq \pchoice{(\prandom{x};\bsprg;\prandom{x})}{(\prandom{x};\csprg;\prandom{x})} & \text{by }\irref{seqdistl}\\
      &\prgeq \pchoice{(\bsprg;\prandom{x})}{(\csprg;\prandom{x})} & \text{by IHs}\\
      &\prgeq (\pchoice{\bsprg}{\csprg});\prandom{x} & \text{by }\irref{seqdistr}
    \end{align*}

    \item $\asprg \equiv \bsprg;\csprg$:
    By assumption, $x \notin\freevars{\bsprg}$ and $x \notin\freevars{\csprg}$ both hold, so we can use both induction hypotheses.
    \begin{align*}
      \prandom{x};\bsprg;\csprg;\prandom{x} &\prgeq \prandom{x};\bsprg;\prandom{x};\csprg;\prandom{x} &\text{by IH}_\csprg\\
      &\prgeq \bsprg;\prandom{x};\csprg;\prandom{x} &\text{by IH}_\bsprg\\
      &\prgeq \bsprg;\csprg;\prandom{x} &\text{by IH}_\csprg
    \end{align*}

    \item $\asprg \equiv \prepeat{\bsprg}$:
    By assumption, $x \notin \freevars{\bsprg}$ so we can derive \((\prandom{x};\bsprg;\prandom{x}) \prgeq \bsprg;\prandom{x}\).
    Intuitively, the idea is to show that the sequence $\bsprg;\bsprg;\dots;\bsprg;\prandom{x}$ is equivalent to $\prandom{x};\bsprg;\prandom{x};\bsprg;\dots;\prandom{x};\bsprg;\prandom{x}$.
    This is formalized by proving the following using the induction hypothesis:

    \begin{calculus}
      \cinferenceRule[aux|aux]{aux}
      {\linferenceRule[eq]
        {\prepeat{\bsprg};\prandom{x}}
        {\prandom{x};\prepeat{(\bsprg;\prandom{x})}}
      }{}    
    \end{calculus}

    The proof of \irref{aux} is done by antisymmetry.
    We first show $\prandom{x};\prepeat{(\bsprg;\prandom{x})} \refines \prepeat{\bsprg};\prandom{x}$.
    \begin{align*}
      \prandom{x};\prepeat{(\bsprg;\prandom{x})} &\prgeq \ptest{\ltrue};\prandom{x};\prepeat{(\bsprg;\prandom{x})} &\text{by }\irref{seqidl}\\
      &\refines (\pchoice{\ptest{\ltrue}}{\bsprg;\prepeat{\bsprg}});\prandom{x};\prepeat{(\bsprg;\prandom{x})} &\text{by }\irref{choiceL}\\
      &\prgeq \prepeat{\bsprg};\prandom{x};\prepeat{(\bsprg;\prandom{x})} &\text{by }\irref{unfold-l}\\
      &\refines \prepeat{\bsprg};\prandom{x} &\text{by }\irref{loopr}\text{ and (A)}
    \end{align*}
    Using \irref{loopr} relies on (A): $\prepeat{\bsprg};\prandom{x};\bsprg;\prandom{x} \refines \prepeat{\bsprg};\prandom{x}$ which holds by using the induction hypothesis.

    Then we show $\prepeat{\bsprg};\prandom{x} \refines \prandom{x};\prepeat{(\bsprg;\prandom{x})}$.
    \begin{align*}
      \prepeat{\bsprg};\prandom{x} &\prgeq \prepeat{\bsprg};\prandom{x};\ptest{\ltrue} &\text{by }\irref{seqidr}\\
      &\refines \prepeat{\bsprg};\prandom{x};(\pchoice{\ptest{\ltrue}}{\bsprg;\prandom{x};\prepeat{(\bsprg;\prandom{x})}}) &\text{by }\irref{choiceL}\\
      &\prgeq \prepeat{\bsprg};\prandom{x};\prepeat{(\bsprg;\prandom{x})} &\text{by }\irref{unfold-l}\\
      &\refines \prandom{x};\prepeat{(\bsprg;\prandom{x})} &\text{by }\irref{loopl}\text{ and (B)}
    \end{align*}
    Using \irref{loopl} relies on (B): $\dbox{\prepeat{\bsprg}}{(\bsprg;\prandom{x};\prepeat{(\bsprg;\prandom{x})} \refines \prandom{x};\prepeat{(\bsprg;\prandom{x})})}$ which holds by \irref{G} and the following refinement:
    \begin{align*}
      \bsprg;\prandom{x};\prepeat{(\bsprg;\prandom{x})} &\refines (\pchoice{\bsprg;\prandom{x}}{\ptest{\ltrue}});\prepeat{(\bsprg;\prandom{x})} &\text{by }\irref{choiceL}\\
      &\prgeq \prepeat{(\bsprg;\prandom{x})} &\text{by }\irref{unfold-l}\\
      &\prgeq {\ptest{\ltrue}};\prepeat{(\bsprg;\prandom{x})} &\text{by }\irref{seqidl}\\
      &\refines \prandom{x};\prepeat{(\bsprg;\prandom{x})} &\text{by }\irref{skip}
    \end{align*}

    Finally, we can finish the proof by using \irref{aux}.
    \begin{align*}
      \prandom{x};\prepeat{\bsprg};\prandom{x} &\prgeq \prandom{x};\prandom{x};\prepeat{(\bsprg;\prandom{x})} &\text{by }\irref{aux}\\
      &\prgeq \prandom{x};\prepeat{(\bsprg;\prandom{x})} &\text{by }\irref{randidem}\\
      &\prgeq \prepeat{\bsprg};\prandom{x} &\text{by }\irref{aux}
    \end{align*}
  \end{itemize}
\end{proof}

With \rref{lem:notfree} proved, we can finally show \rref{thm:fresh}.

\begin{proof}[\rref{thm:fresh}]
  We reason by induction on $C$:
  \begin{itemize}
    \item $C \equiv \cdot_p$:
      \begin{align*}
        \asprg;\prandom{x} &\prgeq \asprg;\prandom{x};\prandom{x} &\text{by }\irref{randidem}
      \end{align*}

    \item $C \equiv \bsprg;C_1$:
      \begin{align*}
        \bsprg;C_1(\asprg);\prandom{x} &\prgeq \bsprg;C_1(\asprg;\prandom{x});\prandom{x} &\text{by IH}
      \end{align*}
    \item $C \equiv C_1;\bsprg$:
      \begin{align*}
        C_1(\asprg);\bsprg;\prandom{x} &\prgeq C_1(\asprg);\prandom{x};\bsprg &\text{by }\irref{refRandomComm}\\
        &\prgeq C_1(\asprg;\prandom{x});\prandom{x};\bsprg &\text{by IH}\\
        &\prgeq C_1(\asprg;\prandom{x});\bsprg;\prandom{x} &\text{by }\irref{refRandomComm}
      \end{align*}    

    \item $C \equiv \pchoice{C_1}{\bsprg}$:
      \begin{align*}
        (\pchoice{C_1(\asprg)}{\bsprg});\prandom{x} &\prgeq \pchoice{(C_1(\asprg);\prandom{x})}{(\bsprg;\prandom{x})} &\text{by }\irref{seqdistr}\\
        &\prgeq \pchoice{(C_1(\asprg;\prandom{x});\prandom{x})}{(\bsprg;\prandom{x})} &\text{by IH}\\
        &\prgeq (\pchoice{C_1(\asprg;\prandom{x})}{\bsprg});\prandom{x} &\text{by }\irref{seqdistr}
      \end{align*}

    \item $C \equiv \pchoice{\bsprg}{C_1}$:
      The proof is identical as the choice operator is commutative.

    \item $C \equiv \prepeat{C_1}$:
      \begin{align*}
        \prepeat{C_1(\asprg)};\prandom{x} &\prgeq \prepeat{(C_1(\asprg);\prandom{x})};\prandom{x} &\text{by (C)}\\
        &\prgeq \prepeat{(C_1(\asprg;\prandom{x});\prandom{x})};\prandom{x} &\text{by IH}\\
        &\prgeq \prepeat{C_1(\asprg;\prandom{x})};\prandom{x} &\text{by (C)}
      \end{align*}

      We prove (C) by antisymmetry, i.e., $\prepeat{\csprg};\prandom{x} \prgeq \prepeat{(\csprg;\prandom{x})};\prandom{x}$ whenever $x \notin\freevars{\csprg}$.
      Note that this condition holds for both $C_1(\asprg)$ and $C_1(\asprg;\prandom{x})$ for which (C) is used.
      The first direction $\prepeat{\csprg};\prandom{x} \refines \prepeat{(\csprg;\prandom{x})};\prandom{x}$ is simple.
      \begin{align*}
        \prepeat{\csprg};\prandom{x} &\refines \prepeat{(\csprg;\ptest{\ltrue})};\prandom{x} &\text{by }\irref{seqidr}\\
        &\refines \prepeat{(\csprg;\prandom{x})};\prandom{x} &\text{by }\irref{skip}
      \end{align*}
      We show the other direction, i.e., $\prepeat{(\csprg;\prandom{x})};\prandom{x} \refines \prepeat{\csprg};\prandom{x}$
      \begin{align*}
        \prepeat{(\csprg;\prandom{x})};\prandom{x} &\prgeq \prepeat{(\csprg;\prandom{x})};\ptest{\ltrue};\prandom{x} &\text{by }\irref{seqidl}\\
        &\refines \prepeat{(\csprg;\prandom{x})};(\pchoice{\ptest{\ltrue}}{\csprg;\prepeat{\csprg}});\prandom{x} &\text{by }\irref{choiceL}\\
        &\prgeq \prepeat{(\csprg;\prandom{x})};\prepeat{\csprg};\prandom{x} &\text{by }\irref{unfold-l}\\
        &\refines \prepeat{\csprg};\prandom{x} &\text{by }\irref{loopl}\text{ and (D)}
      \end{align*}
      Using \irref{loopl} relies on (D): $\dbox{\prepeat{(\csprg;\prandom{x})}}{\csprg;\prandom{x};\prepeat{\csprg};\prandom{x} \refines \prepeat{\csprg};\prandom{x}}$ which holds by \irref{G} and the following refinement:
      \begin{align*}
        (\csprg;\prandom{x};\prepeat{\csprg};\prandom{x}) &\prgeq \csprg;\prepeat{\csprg};\prandom{x} &\text{by \rref{lem:notfree}}\\
        &\refines (\pchoice{\csprg;\prepeat{\csprg}}{\ptest{\ltrue}});\prandom{x} &\text{by }\irref{choiceL}\\
        &\prgeq \prepeat{\csprg};\prandom{x} &\text{by }\irref{unfold-l}
      \end{align*}
  \end{itemize}
\end{proof}
\end{document}